\PassOptionsToClass{accepted=2026-11-03}{quantumarticle}
\documentclass[a4paper,onecolumn,11pt]{quantumarticle}
\pdfoutput=1

\usepackage{silence}
\usepackage[utf8]{inputenc}
\usepackage[T1]{fontenc}
\usepackage[english]{babel}
\usepackage[margin=1in]{geometry}
\usepackage[numbers]{natbib}
\usepackage{amsthm}
\usepackage{amsmath}
\usepackage{latexsym}
\usepackage{amsfonts}
\usepackage{amssymb}
\usepackage{color}
\usepackage{bbm,dsfont}
\usepackage{graphicx}
\usepackage{enumerate}
\usepackage{comment}
\usepackage{adjustbox}
\usepackage{tikz}
\usetikzlibrary{arrows.meta, positioning, shapes, decorations.pathmorphing}
\usepackage{hyperref}
\hypersetup{
    colorlinks=true,
    linkcolor=black,
    breaklinks=true,
    citecolor=blue,
    urlcolor=black
}

\newtheorem{proposition}{Proposition}

\newtheorem{lemma}{Lemma}
\newtheorem{corollary}{Corollary}

\theoremstyle{definition}

\newtheorem{remark}{Remark}
\newtheorem{example}{Example}
\newtheorem{definition}{Definition}

\newcommand{\real}{\mathbb{R}} %real
\newcommand{\complex}{\mathbb{C}} %complex
\newcommand{\nat}{\mathbb N} %natural
\newcommand{\hi}{\mathcal{H}} %Hilbert space
\newcommand{\hik}{\mathcal{K}} %Hilbert space
\newcommand{\lh}{\mathcal{L(H)}} %bounded linear operators
\newcommand{\lhik}{\mathcal{L(K)}} %bounded linear operators
\newcommand{\lsh}{\mathcal{L}_s\mathcal{(H)}} %selfadjoint operators
\newcommand{\lshik}{\mathcal{L}_s\mathcal{(K)}} 
\newcommand{\sh}{\mathcal{S(H)}} %states
\newcommand{\eh}{\mathcal{E(H)}} %effects
\newcommand{\ip}[2]{\left\langle\,#1\,|\,#2\,\right\rangle} %inner product
\newcommand{\kb}[2]{|#1\rangle\langle#2|} %ketbra
\newcommand{\tr}[1]{{\rm tr}\left[#1\right]} %trace
\newcommand{\id}{\mathbbm{1}} %identity operator
\newcommand{\transpose}{\mathsf{T}} %transpose

\newcommand{\rank}[1]{\mathrm{rank}\left(#1\right)}
\newcommand{\nnrank}[1]{\mathrm{rank}_+\left(#1\right)}
\newcommand{\psdrank}[1]{\mathrm{rank}_{psd}\left(#1\right)}
\newcommand{\rspan}[1]{\mathrm{span}_{\mathbb{R}}\left(#1\right)}
\newcommand{\cspan}[1]{\mathrm{span}_{\mathbb{C}}\left(#1\right)}
\newcommand{\lmax}[1]{\lambda_{\max}\left(#1\right)}
\newcommand{\im}{\mathrm{im}}

\begin{document}

\title{Semi-device-independent channel identification with communication matrices}

\author{Samgeeth Puliyil}
\email{samgeethmohanan@gmail.com}
\affiliation{RCQI, Institute of Physics, Slovak Academy of Sciences, D\'ubravsk\'a cesta 9, 84511 Bratislava, Slovakia}
\orcid{0000-0002-3992-4142}
\author{Leevi Lepp\"{a}j\"{a}rvi}
\email{leevi.i.leppajarvi@jyu.fi}
\affiliation{Faculty of Information Technology, University of Jyv\"{a}skyl\"{a}, 40100 Jyv\"{a}skyl\"{a}, Finland}
\orcid{0000-0002-9528-1583}
\author{M\'{a}rio Ziman}
\email{ziman@savba.sk}
\affiliation{RCQI, Institute of Physics, Slovak Academy of Sciences, D\'ubravsk\'a cesta 9, 84511 Bratislava, Slovakia}
\orcid{0000-0002-4805-5890}

\begin{abstract}
  We look into the task of differentiating between any two quantum channels and reconstructing them from the obtained measurement statistics with possibly limited information about the experimental set-up. We employ the communication matrix formalism where the measurement statistics of a prepare-and-measure scenario is represented as a stochastic communication matrix. In order to differentiate between any two quantum channels, the informational completeness of the set-up is both necessary and sufficient. On the other hand, if we want to uniquely characterize any quantum channel, in addition we also need to have a complete description of the set-up. We show that in many important cases we can deduce this information directly from the communication matrix of the set-up before applying the channel. Given that we trust the dimension of the system, we show that we can deduce the information completeness of the set-up directly from the rank of the communication matrix. Furthermore, we show that another quantity of the communication matrix, called the information storability, can be used to self-test the set-up (up to an orthogonal transformation) for an important class of states and measurements. This provides us a semi-device-independent way to identify quantum channels from the prepare-and-measure statistics. Lastly, we consider scenarios where we might have some additional information about the channels or additional resources at our disposal which could help us relax some of the assumptions of the proposed scenario.
\end{abstract}

\maketitle

\section{Introduction}

Identifying an unknown quantum channel is an important task in quantum information processing and experimental set-ups. It is crucial e.g. for characterizing quantum gates in quantum computing setting \cite{Nielsenetal2021}, diagnosing noise and decoherence \cite{EndoBenjaminLi2018,Mangini2022}, benchmarking quantum devices \cite{Chowetal2009} and verifying calibration and error models used for error correction and control \cite{Mangini2024,blumekohout2025easybetterquantumprocess}. Since the set of all quantum channels cannot be perfectly (or unambiguously) discriminated \cite{Pirandola2019}, channel identification requires another set of tools.

One standard way of identifying channels is to perform channel tomography which allows for complete characterization and reconstruction of the transformation. An important primitive of channel tomography is the ability to differentiate between any two quantum channels. Naturally, if we are able to perform channel tomography for a set of channels, then we can also differentiate between them. On the other hand, if we are able to differentiate between any two quantum channels, then in principle we are able to identify any channel by differentiating it from all other (known) channels but one (which is itself). However, unless we have additional information about the channel, performing identification this way is practically impossible because almost certainly we would have to do infinite rounds of differentiation and in each round the other channel would have to be known. Thus, for identifying quantum channels, channel differentiation is a weaker concept than channel tomography.

Even though performing channel tomography enables to differentiate between channels, there are instances where merely differentiating between channels could be more cost-effective. For instance, we could think of a task where we have a device that prepares some fixed quantum channel and we would like to certify that the device prepares the same channel consistently. In this case, although in the first round it might be beneficial to perform channel tomography to know what channel the device is preparing, in the subsequent rounds it is just enough to know that the channel that is being prepared is always the same. This can be beneficial considering that performing full channel tomography can be more costly or time-consuming and it requires more information about the test set-up.

Both channel differentiation and tomography processes can be described in terms of prepare-and-measure scenarios, where one sends some quantum states through the channel and then makes measurements on the transformed states. The distinction of the two concepts becomes clear when considering what type of information one needs to know from the prepare-and-measure set-up in order to be able to accomplish either task. As per the standard formulation \cite{ChuangNielsen1997}, channel tomography is possible whenever one uses informationally complete states and measurement (meaning that the state and measurement operators both span the whole operator space) whose description is completely known. However, just to be able to differentiate between any two channels one does not need the complete information of the prepared states and performed measurement, but just the information about their informational completeness. 

In both cases, however, the information about the set-up needs to be trusted. The central question of our current work is, can we accomplish either of the tasks with less additional trusted information about the set-up and if so, then which type of information we need? We can frame our problem setting as follows: we want to buy an experimental set-up for channel tomography and/or for channel differentiation. The seller assures us that their set-up is able to differentiate between any two channels and/or it can be used for channel tomography for any channel. How can we verify the sellers claims without having to test all possible quantum channels?

We employ the recently introduced communication matrix formalism \cite{Frenkel2015,HeinosaariKerppoLeppäjärvi2020,HeinosaariKerppo2024,HeinosaariKerppoLeppäjärviPlavala2024,leppäjärvi2021measurementsimulabilityincompatibilityquantum} where we collect the measurement outcome statistics of the prepare-and-measure scenario into a stochastic matrix. We show that the informational completeness of the set-up can be readily obtained from the rank of the communication matrix related to the prepare-and-measure scenario before the channel is employed. Once we have verified that the set-up is indeed informationally complete we know that it can be used to differentiate between any two channels. Furthermore, we show that an important class of prepare-and-measure scenarios can be self-tested \cite{Tavakolietal2018,FarkasKaniewski2019,Navascuesetal2023,Miklin2021universalscheme} (up to some gauge) directly from another quantity of the communication matrix called the information storability. Once the set-up has been self-tested it can be readily used to perform channel tomography. This provides a semi-device-independent way to identify any given channel. We then continue to analyze cases where we have some prior information about the channel or some additional resources at our disposal allowing us to relax some of the previous conditions for the set-up. We also look into how to obtain information about some specific properties of channels from their communication matrices.

We note that there are other \emph{calibration-free} channel tomography protocols which also avoid having the full description of the prepare-and-measure set-up, such as gate set tomography \cite{Nielsenetal2021,Merkeletal2013,Greenbaum2015}, shadow process tomography \cite{Kunjummen2023}, and classical certification of quantum gates \cite{Noller2025classical, schroeder2025certifyingquantumgatesautomata}. The goal of our work is to provide a simple self-consistent description in the operationally intuitive communication matrix framework. Our method works in the single-channel level with respect to the unknown prepare-and-measure set-up based on self-testing the set-up or some property of the set-up. We note that the communication-matrix approach should be viewed as complementary to, rather than a replacement for, established channel-tomography techniques. When state preparations and measurements are accurately calibrated, conventional process tomography can exploit this knowledge to obtain efficient quantitative estimates of the channel. Gate set tomography goes further by providing a self-consistent characterization of gates together with state-preparation and measurement operations, making it particularly suitable for diagnosing errors in an entire experimental gate set. Shadow process tomography, on the other hand, can substantially reduce measurement requirements when only selected properties of a process are of interest. Our framework has a different emphasis: it characterizes identifiability from prepare-and-measure statistics under limited assumptions about the devices, and in this sense provides a semi-device-independent operational criterion for channel identification. This comes with corresponding limitations: we do not claim optimal statistical or measurement efficiency, nor do we presently provide the full self-consistent gate-set characterization of GST or the property-specific sample-complexity advantages of shadow methods. The communication-matrix framework is therefore most useful when robustness to imperfect knowledge of the experimental devices and operational identifiability are central considerations.

We start by introducing the necessary mathematical background and the communication matrix formalism for prepare-and-measure scenarios in Sec. \ref{sec:prel}. Next, in Sec. \ref{sec:diff-tomography}, we discuss what properties are required from the prepare-and-measure set-up in order to perform channel differentiation and tomography. In Sec. \ref{sec:sdi-identification} we present our main results about how these properties can be deduced from the communication matrix related to the prepare-and-measure set-up. In Sec. \ref{sec:additional} we explore some scenarios where we might have some additional information about the channels or additional resources at our disposal so that the conditions related to the introduced semi-device-independent scenario can be relaxed in order to infer some weaker properties of the set-up. Finally in Sec. \ref{sec:concl} we present our final remarks.

\section{Communication matrices and prepapre-and-measure scenarios}\label{sec:prel}

In an operational theory, communication tasks between two parties often involve one of them (say, Alice) sending an encoded message (state) to the other (say, Bob), who decodes the message using a measurement. Such schemes are called {\it prepare-and-measure scenarios}. Here, we recall such a prepare-and-measure scenario discussed in \cite{Frenkel2015,HeinosaariKerppoLeppäjärvi2020,HeinosaariKerppo2024,HeinosaariKerppoLeppäjärviPlavala2024,leppäjärvi2021measurementsimulabilityincompatibilityquantum}, where a simple yet useful tool of communication matrices in operational theories was introduced. We will only look at the standard operational quantum theory for the time being.

We denote by $\mathcal{L}(\mathcal{H})$ the set of all bounded linear operators acting on a $d$-dimensional Hilbert space $\mathcal{H}$, and by $\mathcal{L}_s(\mathcal{H})$ the set of all bounded linear self-adjoint operators on $\mathcal{H}$. Then the state space on $\hi$ is the set $\mathcal{S}(\mathcal{H})$ of all \emph{density operators} on $\hi$ (positive semidefinite operators with unit trace). The set of \emph{effects}, i.e. the set of selfadjoint operators $E$ between the zero operator $O$ and the identity operator $\id_d$ on $\hi$, i.e., $O \leq E \leq \id_d$  is denoted by $\mathcal{E}(\mathcal{H})$. A measurement on $\hi$ is then described by a \emph{positive operator-valued measure (POVM)} $M$ which in the case of finite outcomes, say $n < \infty$ outcomes, is described by a collection of effects $\{M_k \}_{k=1}^n \subset \eh$ such that $\sum_{k=1}^n M_k = \id_d$. We denote the set of POVMs on $\hi$ with $n$ outcomes by $\mathcal{M}_n(\hi)$. By $\mathcal{P}(\mathcal{H})$ we denote the set of all projections (operators $P \in \lh$ with $P^2=P^*=P$, where $P^*$ is the adjoint of $P$). We also denote the set of positive integers from 1 to $n$, $\{1,...,n\}$, by $[n]$. 

Our prepare-and-measure scenario is the following: Alice prepares a state from a set of some $m$ states $\{\varrho_j\}_{j\in[m]}\subset\mathcal{S}(\mathcal{H})$ and sends it to Bob, who performs a fixed measurement $M \in \mathcal{M}_n(\hi)$ with $n$ outcomes on the state. The process is repeated until we get the measurement statistics given by the Born rule, which are used to construct an object called the \textit{communication matrix} (or channel matrix) $C= (C_{jk})_{j,k} = (\tr{\varrho_jM_k})_{j,k}$, which is a non-negative row-stochastic matrix \cite{Frenkel2015,HeinosaariKerppoLeppäjärvi2020}. A schematic diagram of the process is shown in Fig.\ref{figure1}. Given a set of states $\{\varrho_j\}_{j=1}^m$ and a measurement $M$ with $n$ outcomes, the induced communication matrix is denoted by $C_{\varrho, M}$. We note that for a given communication matrix the implementation might not be unique (we come back to this point in Sec. \ref{sec:sdi-identification}). The set of all $m\times n$ communication matrices that can be implemented with a $d$-dimensional quantum system is denoted by $\mathcal{C}_{m,n}(\mathcal{H})$, and the set of all communication matrices that can be implemented with a $d$-dimensional quantum system is denoted by $\mathcal{C}(\mathcal{H}) = \bigcup_{m,n}\mathcal{C}_{m,n}(\mathcal{H})$.

We can interpret any such prepare-and-measure scenario as some kind of a communication task between Alice and Bob. Then a communication matrix captures the essence of the communication task described by this prepare-and-measure scenario. Although the implementation of a given communication matrix might not be unique, the communication task itself can be identified from the communication matrix. For example, the identity matrix $I_n$ corresponds to a task of \emph{perfectly distinguishing} $n$ states \cite{Barnett2009,Bae_2015}. Likewise, the matrix $A_n$ with $(A_n)_{jj}= 0$ and $(A_n)_{jk}= 1/(n-1)$ for $j, k \in [n]$, where $j \neq k$, corresponds to the task of uniformly \emph{antidistinguishing (excluding)} $n$ states \cite{Caves2002Conditions,Pusey2012On,Bandyopadhyay2014Conclusive,Heinosaari2018Antidistinguishability,Leifer2020Noncontextuality,Havlicek2020Simple,russo2023inner}. Between the two extermes $I_n$ and $A_n$ we can also introduce the \emph{noisy uniform antidistinguishability communication matrices} $D_{n, \epsilon}$ defined as $(D_n)_{jj} = 1- \epsilon$ and $(D_n)_{jk} = \epsilon/(n-1)$ for $j, k \in [n]$, where $j \neq k$, where $\epsilon \in [0,1]$ describes noise in the (anti)distinguishability task. 

In general, set of implementable communication matrices capture many of the important characteristics of a theory. For example in qubit we know that we can implement $I_n$ if and only if $n\leq 2$ (this reflects the fact that we can only distinguish two orthogonal pure states in qubit at a time), $A_n$ if and only if $n \leq 4$ (we can only distribute $4$ points on the Bloch sphere with equal distances) and $D_{n,\epsilon}$ if and only if $n\leq 4$ and $\epsilon \in [2/n,1]$ (the basic decoding theorem \cite{SchumaherWestmorelandbook}).

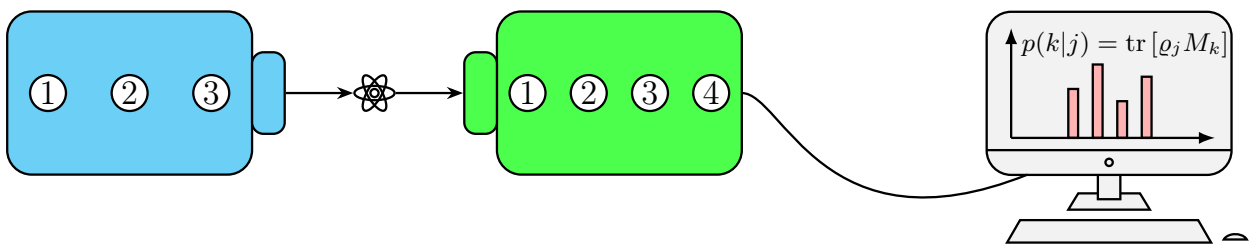
\begin{figure}[ht]
\centering
\begin{adjustbox}{width=\textwidth}
\begin{tikzpicture}[
    box/.style={rounded corners=3mm, draw, minimum width=1cm, minimum height=1cm},
    jar/.style={rounded corners=3mm, draw, minimum width=3cm, minimum height=2cm},
    jarcap/.style={rounded corners, draw, minimum width=0.4cm, minimum height=1cm},
    thick, numcircle/.style={circle, draw, minimum size=4mm, inner sep=1pt},
    nucleus/.style={circle, draw, minimum size=1.5pt, inner sep=1pt},
    orbit1/.style={ellipse, draw, rotate=0, minimum width=14pt, minimum height=5pt, inner sep=1pt},
    orbit2/.style={ellipse, draw, rotate=60, minimum width=14pt, minimum height=5pt, inner sep=1pt},
    orbit3/.style={ellipse, draw, rotate=120, minimum width=14pt, minimum height=5pt, inner sep=1pt},
    computer/.style={rounded corners=3mm, draw, minimum width=3cm, minimum height=2cm},
    stand/.style={rectangle, draw, minimum width=2mm, minimum height=3mm},
    standbase/.style={trapezium, draw, trapezium stretches body=true, trapezium angle=30, minimum height=2mm, minimum width=1cm, inner sep=1pt},
    keyboard/.style={trapezium, draw, trapezium stretches body=true, trapezium angle=60, minimum width=2.5cm, minimum height=0.15cm},
    mouse/.style={rotate=90, scale=0.6}
]

% Left jar with 4 circles
\node[thick, jar, fill=cyan!50] (leftjar) at (-6,0) {};
\node[thick, jarcap, fill=cyan!50] (leftjarcap) at (-4.3,0) {};
\foreach \i in {1,...,3}
    \node[thick, numcircle, fill=white!100] (left\i) at ({-6 + (\i - 2)}, 0) {\i};

% Right jar with 3 circles
\node[thick, jar, fill=green!70] (rightjar) at (0,0) {};
\node[thick, jarcap, fill=green!70] (rightjarcap) at (-1.7,0) {};
\foreach \i in {1,...,4}
    \node[thick, numcircle, fill=white!100] (right\i) at ({(\i - 2.5)*0.75}, 0) {\i};

% Atomic symbol connection
\node[thick, nucleus] (nucleus) at (-3,0) {};
\node[thick, orbit1] (orbit1) at (-3,0) {};
\node[thick, orbit2] (orbit2) at (-3,0) {};
\node[thick, orbit3] (orbit3) at (-3,0) {};

% Connection between jars
\draw[thick, -{Stealth[length=2mm]}] (leftjarcap.east) -- (orbit1.west);
\draw[thick, -{Stealth[length=2mm]}] (orbit1.east) -- (rightjarcap.west);

% Computer system
\node[thick, standbase, fill=black!5] (standbase) at (6,-1.33) {};
\node[thick, above=-1mm of standbase, stand, fill=black!5] (stand) {};
\node[thick, above=-0.3mm of stand, computer, fill=black!5] (monitor) {};
\node[thick, below=0.1cm of standbase, keyboard, fill=black!5] (keyboard) {};
\draw[thick, fill=black!5] ([xshift=5mm, yshift=-1mm]keyboard.east) arc (20:160:0.15cm and 0.1cm) -- cycle;
\draw[thick] (4.5,-0.7) -- (7.5,-0.7);
\node[thick, circle, draw, inner sep=0.9pt] (button) at (6,-0.85) {};

% Graph in monitor
\begin{scope}[xshift=5.5cm, yshift=-0.55cm, scale=0.3]
    \foreach \x/\h in {0/2, 1/3, 2/1.5, 3/2.5}
        \draw[fill=red!30] (\x,0) rectangle (\x+0.4,\h);
\end{scope}

\node[] at (6.2,0.6){\scriptsize$p(k|j)=\tr{\varrho_jM_k}$};

\draw[thick, -{Latex[length=2mm]}] (4.8,-0.55) -- (7.3,-0.55);
\draw[thick, -{Latex[length=2mm]}] (4.8,-0.55) -- (4.8,0.8);

\draw[densely dotted] (1.5,0) .. controls (2,0) and (2.5,-2) .. (5,-1);

\end{tikzpicture}
\end{adjustbox}
\caption{A schematic diagram for the prepare-and-measure scenario.}\label{figure1}
\end{figure}

Communication tasks such as the one mentioned above use information \emph{channels} to send messages from Alice to Bob. In quantum theory, a channel $\Phi: \lh \to \mathcal{L(K)}$ between Hilbert space $\hi$ and $\mathcal{K}$ is given by a completely positive (CP) ($\Phi \otimes id_r$ is positive map for all $r \in \nat$) and trace-preserving (TP) ($\tr{\Phi(X)} =\tr{X}$ for all $X \in \lh$) maps. We denote the set of CPTP maps $\Phi: \lh \to \hik$ by $\mathrm{Ch}(\hi,\hik)$. In general, the transformation using a CPTP map can be thought of as part of the state preparation (Schr{\"o}dinger picture) or as part of the measurement (Heisenberg picture), which makes the task equivalent to the procedure mentioned above. However, what happens when the object of interest is the transformation itself?

In order to capture the transformations explicitly, we propose a modification to the previous scenario: Let $\mathcal{H}$ and $\mathcal{K}$ be two Hilbert spaces of dimensions $d_i$ and $d_o$, respectively. Alice prepares a quantum state from a given set of states $\{\varrho_j\}_{j\in[m]}\subset\mathcal{S}(\mathcal{H})$, which is sent to Bob through a quantum channel $\Phi:\mathcal{L}(\mathcal{H})\to\mathcal{L}(\mathcal{K})$. Bob performs a fixed measurement $M \in \mathcal{M}_n(\hik)$ on the state. A schematic diagram is shown in Fig. \ref{figure2}.

In this modified prepare-and-measure scheme, the communication matrix is given by
\begin{equation}\label{eq3.2}
    C' = (C'_{jk})_{j,k} = (\tr{\Phi(\varrho_j)M_k})_{j,k} \, .
\end{equation}
The set of all $m\times n$ communication matrices that can be constructed in the presence of some channel $\Phi:\mathcal{L}(\mathcal{H})\to\mathcal{L}(\mathcal{K})$ is given by $\mathcal{C}_{m,n}(\mathcal{H},\mathcal{K})$, and the set of all communication matrices that can be constructed in the presence of some channel $\Phi:\mathcal{L}(\mathcal{H})\to\mathcal{L}(\mathcal{K})$ is given by $\mathcal{C}(\mathcal{H},\mathcal{K}) = \bigcup_{m,n}\mathcal{C}_{m,n}(\mathcal{H},\mathcal{K})$. 

\begin{figure}[ht]
\centering
\begin{adjustbox}{width=\textwidth}
\begin{tikzpicture}[
    box/.style={rounded corners=3mm, draw, minimum width=1cm, minimum height=1cm},
    jar/.style={rounded corners=3mm, draw, minimum width=3cm, minimum height=2cm},
    jarcap/.style={rounded corners, draw, minimum width=0.4cm, minimum height=1cm},
    thick, numcircle/.style={circle, draw, minimum size=4mm, inner sep=1pt},
    nucleus/.style={circle, draw, minimum size=1.5pt, inner sep=1pt},
    orbit1/.style={ellipse, draw, rotate=0, minimum width=14pt, minimum height=5pt, inner sep=1pt},
    orbit2/.style={ellipse, draw, rotate=60, minimum width=14pt, minimum height=5pt, inner sep=1pt},
    orbit3/.style={ellipse, draw, rotate=120, minimum width=14pt, minimum height=5pt, inner sep=1pt},
    computer/.style={rounded corners=3mm, draw, minimum width=3cm, minimum height=2cm},
    stand/.style={rectangle, draw, minimum width=2mm, minimum height=3mm},
    standbase/.style={trapezium, draw, trapezium stretches body=true, trapezium angle=30, minimum height=2mm, minimum width=1cm, inner sep=1pt},
    keyboard/.style={trapezium, draw, trapezium stretches body=true, trapezium angle=60, minimum width=2.5cm, minimum height=0.15cm},
    mouse/.style={rotate=90, scale=0.6}
]

% Left jar with 4 circles
\node[thick, jar, fill=cyan!50] (leftjar) at (-6,0-3) {};
\node[thick, jarcap, fill=cyan!50] (leftjarcap) at (-4.3,0-3) {};
\foreach \i in {1,...,3}
    \node[thick, numcircle, fill=white!100] (left\i) at ({-6 + (\i - 2)}, 0-3) {\i};

% Right jar with 3 circles
\node[thick, jar, fill=green!70] (rightjar) at (0,0-3) {};
\node[thick, jarcap, fill=green!70] (rightjarcap) at (-1.7,0-3) {};
\foreach \i in {1,...,4}
    \node[thick, numcircle, fill=white!100] (right\i) at ({(\i - 2.5)*0.75}, 0-3) {\i};

\node[thick, box, fill=black!15] (channel) at (-3,0-3) {$\Phi$};

% Connection between jars
\draw[thick, -{Stealth[length=2mm]}] (leftjarcap.east) -- (channel.west);
\draw[thick, -{Stealth[length=2mm]}] (channel.east) -- (rightjarcap.west);

% Computer system
\node[thick, standbase, fill=black!5] (standbase) at (6,-1.33-3) {};
\node[thick, above=-1mm of standbase, stand, fill=black!5] (stand) {};
\node[thick, above=-0.3mm of stand, computer, fill=black!5] (monitor) {};
\node[thick, below=0.1cm of standbase, keyboard, fill=black!5] (keyboard) {};
\draw[thick, fill=black!5] ([xshift=5mm, yshift=-1mm]keyboard.east) arc (20:160:0.15cm and 0.1cm) -- cycle;
\draw[thick] (4.5,-0.7-3) -- (7.5,-0.7-3);
\node[thick, circle, draw, inner sep=0.9pt] (button) at (6,-0.85-3) {};

% Graph in monitor
\begin{scope}[xshift=5.5cm, yshift=-3.55cm, scale=0.3]
    \foreach \x/\h in {0/1.5, 1/2.5, 2/1, 3/3}
        \draw[fill=red!30] (\x,0) rectangle (\x+0.4,\h);
\end{scope}

\node[] at (6.2,0.6-3){\scriptsize$p'(k|j)=\tr{\varrho'_jM_k}$};

\draw[thick, -{Latex[length=2mm]}] (4.8,-0.55-3) -- (7.3,-0.55-3);
\draw[thick, -{Latex[length=2mm]}] (4.8,-0.55-3) -- (4.8,0.8-3);

\draw[densely dotted] (1.5,0-3) .. controls (2,0-3) and (2.5,-2-3) .. (5,-1-3);

\end{tikzpicture}
\end{adjustbox}
\caption{A schematic diagram for the prepare-and-measure scenario in the presence of a channel.}\label{figure2}
\end{figure}
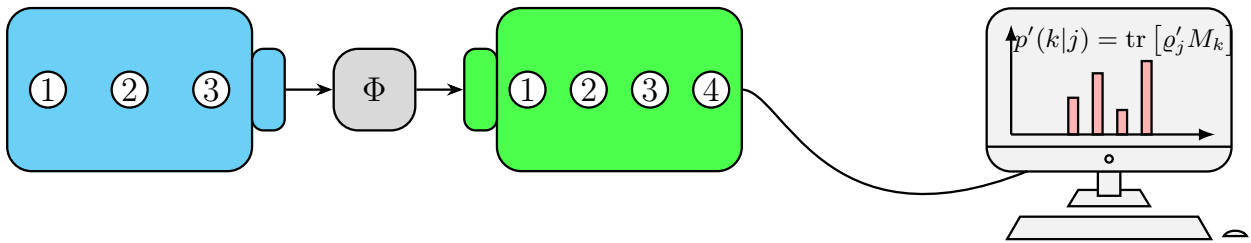

In \cite{HeinosaariKerppoLeppäjärvi2020} the authors introduced some mathematical quantities on the communication matrices, which can be used to infer some properties of the theory.  In Secs. \ref{sec:sdi-identification} and \ref{sec:additional} we discuss in more detail these quantities and how these quantities can be used for identifying channels or some properties of the channels.

\section{Channel differentiation and tomography}\label{sec:diff-tomography}

Here we discuss the necessary and sufficient conditions on the states and the measurement used in the prepare-and-measure scenario depicted in Fig. \ref{figure2} to differentiate between any two quantum channels, and to perform complete channel tomography, to show that the two tasks are not synonymous. We further investigate both tasks in the case where the channel is already known to be unital.

\subsection{Differentiability of channels}\label{subsec:differentiation}

Let us first focus on differentiating channels. We start by a formal definition.
\begin{definition}
    A prepare-and-measure scenario consisting of states $\{\varrho_j\}_{j\in[m]}\subset\mathcal{S}(\mathcal{H})$ and a measurement $M \in \mathcal{M}_n(\hik)$ can \emph{differentiate} between two quantum channels $\Phi_1, \Phi_2: \lh \to \lhik$ if and only if $C_{\Phi_1(\varrho),M} = C_{\Phi_2(\varrho),M}$ implies that $\Phi_1 = \Phi_2$.
\end{definition}

For a set of states $\{\varrho_j\}_{j\in[m]}\subset\mathcal{S}(\mathcal{H})$ and a measurement $M \in \mathcal{M}_n(\hik)$ we denote $V_\varrho= \rspan{\{\varrho\}_{j\in [m]}} \subseteq \lsh$ and $V_M = \rspan{\{M_k\}_{k \in [n]}} \subseteq \lshik$. It is no surprise that differentiation of any two channels between $\hi $ and $\hik$ is possible only when the states and the measurement are \emph{informationally complete}, i.e., when $V_\varrho = \lsh$ and $V_M = \lshik$. While it is an elementary observation that informational completeness of the states and the measurement is sufficient condition for differentiating between any two channels, our proof presents a construction where for any informationally incomplete set of states or measurement we can define two channels which cannot be differentiated.

Before going to the formal statement, let us recall an elementary property of quantum state spaces discussed in \cite{Bengtsson_Zyczkowski_2006, Bertlmann_2008}. We give a proof for completeness and to suit our chosen (scaling of the) parametrization for quantum states.

\begin{lemma}\label{lemma1}
        The radius of the biggest $d^2-1$ sphere that is completely contained within the set of Bloch vectors of all $d$-dimensional quantum states is $\frac{1}{\sqrt{d-1}}$.
\end{lemma}

\begin{proof}
        A quantum state is a positive-semidefinite unit trace operator acting on a Hilbert space $\mathcal{H}$. Using a basis consisting of the identity operator $\id$ and a set of $d^2-1$ traceless self-adjoint operators on $\hi$ such that $\tr{\sigma_j\sigma_k} = d\delta_{jk}$, a d-dimensional quantum state can be written in the form $\varrho = \frac{1}{d}(\id_d + \vec{r} \cdot\vec{\sigma})$, where $(\vec{\sigma})_j = \sigma_j$ and $\vec{r}\in \real^{d^2-1}$. Furthermore, since the purity of the density operator satisfies the condition $\tr{\varrho^2}\leq1$, it is straightforward to see that $\frac{1}{d}(1 + ||\vec{r}||^2)\leq1$, or in other words, $||\vec{r}|| \leq\sqrt{d-1}$.
        While any quantum state can be written in the above form, not all operators of this form are positive-semidefinite. Denote $\vec{r} = R\hat{r}$, where $R:=||\vec{r}||\geq0$ and $\hat{r}$ is the unit vector along $\vec{r}$. An operator of the form $\varrho = \frac{1}{d}(\id_d + \vec{r} \cdot\vec{\sigma})$ is positive-semidefinite if and only if it has only non-negative eigenvalues. That is,
        \begin{equation}
        \begin{split}
            1 + R\mu_i &\geq 0 \qquad\forall\;i\\
            \implies R\mu_i &\geq -1 \quad\forall\;i,\label{radineq0}
        \end{split}
        \end{equation}
        where $\mu_i$ are the eigenvalues of $\hat{r}\cdot\vec{\sigma}$. Eq.\eqref{radineq0} holds for all $i \in [d]$ if it holds for the smallest eigenvalue, $\mu_{\min}$. Since $\tr{\vec{r}\cdot\vec{\sigma}} = R\tr{\hat{r}\cdot\vec{\sigma}} = R\sum_{i=1}^{d}\mu_i = 0$, the smallest eigenvalue of $\hat{r}\cdot\vec{\sigma}$ is negative $(\mu_{\min}<0)$. Therefore, any operator of the form $\varrho = \frac{1}{d}(\id_d + R \hat{r} \cdot\vec{\sigma})$ is positive-semidefinite if and only if
        \begin{equation}\label{radineq}
            R|\mu_{\min}| \leq 1 \, .
        \end{equation}
        Since $\hat{r}$ is a unit vector, and since $\tr{\sigma_j}=0$ and $\tr{\sigma_j^2}=d$ for all $j\in[d^2-1]$, we have $\tr{\hat{r}\cdot\vec{\sigma}} = 0$ and $\tr{(\hat{r}\cdot\vec{\sigma})^2} = d$. In terms of the eigenvalues $\{\mu_i\}_{i\in[d]}$, this becomes
        \begin{equation}
            \sum_{i=1}^d\mu_i = 0,\quad\sum_{i=1}^d\mu^2_i = d.
        \end{equation}
        Without loss of generality, consider $\mu_{\min}=\mu_1=-a$, where $a>0$. Which gives us
        \begin{equation}
            \sum_{i=2}^d\mu_i = a,\quad\sum_{i=2}^d\mu^2_i = d - a^2.
        \end{equation}
        Using the Cauchy-Schwarz inequality, we get
        \begin{align}
            a^2 = \left(\sum_{i=2}^d\mu_i\right)^2 &\leq (d-1)\left(\sum_{i=2}^d\mu^2_i\right) = (d-1)(d-a^2),\\
            \implies a^2 &\leq d(d-1) - a^2(d-1),\\
            \implies da^2 &\leq d(d-1),\\
            \implies a &\leq \sqrt{d-1}.\label{csineq}
        \end{align}
        Therefore, $|\mu_{\min}| \leq \sqrt{d-1}$ in all directions $\hat{r}$. Therefore, for all $R\leq\frac{1}{\sqrt{d-1}}$, Eq. \eqref{radineq} is satisfied and $\varrho = \frac{1}{d}(\id_d + \vec{r} \cdot\vec{\sigma})$ is a quantum state.
        
        On the other hand, the equality in Eq. \eqref{csineq} is achieved only when $\mu_{\min}=-\sqrt{d-1}$ and $\mu_{i} = \frac{1}{\sqrt{d-1}}$ for all $i\in\{2,\dots,d\}$. Define $H = \text{diag}(\mu_i), i\in[d]$. It is clear that $\tr{H} = \sum_{i=1}^d\mu_i = 0$ and $\tr{H^2} = \sum_{i=1}^d\mu_i^2 = d$. Therefore, $H = \hat{s}\cdot\vec{\sigma}$ for some unit vector $\hat{s}$. This means that for any $R>\frac{1}{\sqrt{d-1}}$, we have that $R|\mu_{\min}| = R(\sqrt{d-1}) > 1$, which means that the operator $\varrho = \frac{1}{d}(\id_d + R\hat{s} \cdot\vec{\sigma})$ is not positive semidefinite, and thus not a quantum state. This completes the proof of the lemma.
\end{proof}

With this, we move on to the formal statement regarding differentiability of quantum channels.

\begin{proposition}\label{prop:differentiability}
    A communication matrix $C_{\varrho,M}$ with an implementation with states $\{\varrho_j\}_{j\in[m]}$ $\subset\mathcal{S}(\mathcal{H})$ and a measurement $M \in \mathcal{M}_n(\hik)$ can differentiate between any two quantum channels $\Phi_1,\Phi_2:\mathcal{L}(\mathcal{H})\to\mathcal{L}(\mathcal{K})$ if and only if both the states and the measurement are informationally complete.
\end{proposition}
\begin{proof}
    $(\impliedby)$ Let us start with a set of states $\{\varrho_j\}_{j \in [m]}$ and a measurement $M \in \mathcal{M}_n(\hik)$ with effects $\{M_k\}_{k \in [n]}$ that span $\mathcal{L}_s(\mathcal{H})$ and $\mathcal{L}_s(\mathcal{K})$ over the field of real numbers, respectively. This also means that the states $\{\varrho_j\}_{j \in [m]}$ and the effects $\{M_k\}_{k \in [n]}$ span $\mathcal{L}(\mathcal{H})$ and $\mathcal{L}(\mathcal{K})$ over the field of complex numbers, respectively. Now, let us consider two quantum channels $\Phi_1,\Phi_2:\mathcal{L}(\mathcal{H})\to\mathcal{L}(\mathcal{K})$. Suppose that $C_{\Phi_1(\varrho),M} = C_{\Phi_2(\varrho),M}$. This means that $\tr{\Phi_1(\varrho_j)M_k} = \tr{\Phi_2(\varrho_j)M_k}$ for all $j \in [m]$ and $k \in [n]$. Since the effects $\{M_k\}_{k \in [n]}$ span the whole $\lhik$, we must have that $\Phi_1(\varrho_j) = \Phi_2(\varrho_j)$ for all $j \in [m]$.
    Since $\{\varrho_j\}_{j\in [n]}$ spans $\mathcal{L}(\mathcal{H})$ we must have that $\Phi_1 = \Phi_2$. Thus, the set of states $\{\varrho_j\}_j$ and the measurement $M$ can differentiate between any two channels from $\hi$ to $\hik$.

    $(\implies)$
    Denote $\dim(\mathcal{H})=d_{i}$ and $\dim(\mathcal{K})=d_{o}$ and the identity operators acting on $\mathcal{H}$ and $\mathcal{K}$ as $\id_{d_i}$ and $\id_{d_o}$, respectively.
    
    \textit{Case 1:} Suppose that the set of states $\{\varrho_j\}_{j \in [m]}$ does not span $\mathcal{L}_s(\mathcal{H})$. That is, $\dim(V_\varrho)\leq d_{i}^2-1$. This means that there exists $A\in\mathcal{L}_s(\mathcal{H})$ such that $\tr{A\varrho_j} = 0$ for all $j \in [m]$. Since $A = A^*$, we can write
    \begin{equation}
        A = \alpha \id_{d_i} + A_0
    \end{equation}
    where $\alpha = \frac{\tr{A}}{d_{i}}$ is real and $A_0$ is a traceless self-adjoint operator. We can normalize $A_0$ and define $B_1 = \frac{A_0}{\beta}$, where $\beta = \sqrt{\frac{\tr{A_0^2}}{d_{i}}}$ is real and positive. $B_1$ is a traceless self-adjoint operator such that $\tr{B_1^2} = d_{i} = \tr{\sigma_1^2}$, where $\sigma_1$ is one of the $d$-dimensional Bloch basis elements for $\mathcal{L}_s(\mathcal{H})$. Thus, there exists a unitary operator $U$ such that $B_1 = U\sigma_1U^*$. We use the same unitary on all $\sigma_u$ to create a new set of traceless self-adjoint operators $\{B_u = U\sigma_uU^*\}_{u\in[d_{i}^2-1]}$ that, along with the identity operator, form an orthogonal operator basis ($\tr{B_uB_v} = \delta_{uv}d_{i}$). Any quantum state can now be written in terms of the new operator basis as
    \begin{equation}
        \varrho = \frac{1}{d_{i}}\left(\id_{d_i} + \sum_{u=1}^{d_{i}^2-1} s_uB_u\right),
    \end{equation}
    which includes all the $\varrho_j$. The trace between A and $\varrho_j$ for all $j \in [m]$ then gives us the following:
    \begin{equation}
        \begin{split}
            0&=\tr{A\varrho_j} =\tr{(\alpha \id_{d_i} + \beta B_1)\varrho_j} = \alpha \tr{\varrho_j}+ \beta\tr{B_1\varrho_j} \\
            &= \alpha + \beta\tr{\frac{B_1}{d_{i}}\left(\id_{d_i} + \sum_{u=1}^{d^2-1} s_u^{(j)}B_u\right)} =\alpha + \beta s_1^{(j)}, 
        \end{split}
    \end{equation}
    which shows that $s_1^{(j)} = -\frac{\alpha}{\beta}$ for all $j \in [m]$. This means that all $\varrho_j$ have the same component along $B_1$. Define two operators $N_\pm = \frac{1}{2}(\id_{d_i} \pm tB_1)$, where $|t|\leq\frac{1}{\sqrt{d_{i}-1}}$. Clearly, $N_\pm\in\mathcal{L}_s(\mathcal{H})$, and for any operator $\varrho = \frac{1}{d}(\id_{d_i} + \sum_u s_u B_u)$ that resides in the $d_{i}^2-1$ dimensional Bloch ball, so that $||s|| \leq \sqrt{d_i-1}$, we have that
    \begin{equation}
        \tr{N_\pm\varrho} = \frac{1}{2}(1 \pm ts_1).
    \end{equation}
    Since $|t|\leq\frac{1}{\sqrt{d_{i}-1}}$ and $|s_1|\leq\sqrt{d_{i}-1}$ for any vector inside the $d_{i}^2-1$ dimensional Bloch ball, we have $|ts_1|\leq1$, that is, $\tr{N_\pm\varrho} \geq 0$. Since all density matrices reside in the $d_{i}^2-1$ dimensional Bloch ball and since states form a base for the (self dual) cone of positive semidefinite operators on $\hi$, we have that $N_\pm\geq0$. Additionally, $N_+ + N_- = \id_{d_i}$, which makes $N$ consisting of effects $\{N_\pm\}$ a dichotomic measurement on $\mathcal{S}(\mathcal{H})$. Denote $p = \tr{N_+\varrho_j} = \frac{1}{2}(1 - \frac{t\alpha}{\beta})$ and $(1-p) = \tr{N_-\varrho_j} = \frac{1}{2}(1 + \frac{t\alpha}{\beta})$, which are the same for all $j \in [m]$. Now, fix two states $\xi_1,\xi_2\in\mathcal{S}(\mathcal{K})$, and define two channels $\Phi_1,\Phi_2:\mathcal{L}(\mathcal{H})\to\mathcal{L}(\mathcal{K})$ as follows: 
    \begin{align}
        \Phi_1(X) &= \tr{X} (p\xi_1 + (1-p)\xi_2),\\
        \Phi_2(X) &= \tr{N_+X}\xi_1 + \tr{N_-X}\xi_2,
    \end{align}
    for all $X \in \lh$. It can be easily seen that $\Phi_1(\varrho_j) = \Phi_2(\varrho_j)$ for all $j \in [m]$, even though $\Phi_1\neq\Phi_2$. Thus, $\Phi_1$ and $\Phi_2$ cannot be differentiated using $\{\varrho_j\}$ and any measurement. Thus, in order to differentiate between any two channels, the states need to be informationally complete.
    
    \textit{Case 2:} Suppose that the measurement $M$ is not informationally complete. That is, there exists $A_0\in\mathcal{L}_s(\mathcal{K})$ such that $\tr{A_0M_k}=0$ for all $k \in [n]$. This means that $\tr{A_0} = \tr{A_0(\sum_kM_k)}=0$, that is, $A_0$ is a traceless, self-adjoint operator. Similar to \textit{Case 1}, we can normalize $B_1$, that is, $B_1=\frac{A_0}{\beta}$, where $\beta = \sqrt{\frac{\tr{A_0^2}}{d_{o}}}$ is real and positive. This gives us $\tr{B_1M_k} = 0$ for all $k \in [n]$ and $\tr{B_1^2}=d_{o}=\tr{\sigma_1^2}$, where $\sigma_1$ is one of the $d$-dimensional Bloch basis elements for $\mathcal{L}_s(\mathcal{K})$, implying the existence of a unitary such that $B_1 = U\sigma_1U^*$. We can define a new Bloch basis similar to how we did in \textit{Case 1} as $\{B_u = U\sigma_uU^*\}_{u\in[d_{o}^2-1]}$ and the identity operator $\id_{d_o}$. Any state $\varrho\in\mathcal{S}(\mathcal{K})$ that can be written as $\varrho = \frac{1}{d_{o}}(\id_{d_o} + r_1B_1)$ gives us,

    \begin{equation}
    \begin{split}
        \tr{\varrho M_k} &= \frac{1}{d_{o}}\tr{\left(\id_{d_o}+r_1B_1\right)M_k}\\
        &= \frac{\tr{M_k}}{d_{o}}\qquad\forall\;k \in [n]\, .
    \end{split}        
    \end{equation}

    Let us now define two channels $\Phi_1,\Phi_2:\mathcal{L}(\mathcal{H})\to\mathcal{L}(\mathcal{K})$ as follows:
    \begin{align}
        \Phi_1(X) &= \tr{X}\frac{\id_{d_o}}{d_{o}},\\
        \Phi_2(X) &= \tr{N_+X}\zeta_+ + \tr{N_-X}\zeta_-
    \end{align}
    for all $X \in \lh$, where $N$ with effects $\{N_\pm\}$ is the same measurement as defined in \textit{Case 1} and $\zeta_\pm = \frac{1}{d_{o}}(\id_{d_o} \pm r_1B_1)$, where $|r_1|\leq\frac{1}{\sqrt{d_{o}-1}}$, meaning that $\zeta_\pm\in\mathcal{S}(\mathcal{K})$ according to Lemma \ref{lemma1}. It can be seen that for all $k \in [n]$ and $X \in \lh$, we have the following:

    \begin{equation}
    \begin{split}
        \tr{\Phi_2(X)M_k} &= \tr{\left(\tr{N_+X}\zeta_+ + \tr{N_-X}\zeta_-\right)M_k}\\
        &= \frac{1}{d_{o}}\tr{\left(\tr{X}\id_{d_o} + \tr{(N_+ - N_-)X}r_1B_1\right)M_k}\\
        &= \frac{\tr{X}\tr{M_k}}{d_{o}}\\
        &= \tr{\Phi_1(X)M_k},
    \end{split}  
    \end{equation}
    which is also true for all $X\in\{\varrho_j\}_{j \in [m]}$. That is, $C_{\Phi_1(\varrho),M} = C_{\Phi_2(\varrho),M}$, which means that $\Phi_1$ and $\Phi_2$ cannot be differentiated using $M$. Thus, any two quantum channels $\Phi_1,\Phi_2:\mathcal{L}(\mathcal{H})\to\mathcal{L}(\mathcal{K})$ can be differentiated only if $M$ is informationally complete. This completes the proof.
\end{proof}

\subsection{Channel tomography}\label{subsec:tomography}

Next we look into how to perform channel tomography in the prepare-and-measure scenario so that the channel can be reconstructed from the elements of the communication matrix. In order to reconstruct a channel from the prepare-and-measure statistics one must be able to solve all the channel parameters related to some parametrization. In the following we will use the Bloch parametrization of a channel obtained from some Bloch basis of operators (see e.g. \cite{heinosaari2011mathematical}). One could equally work for example with the $\chi$-matrix representation obtained from (some) Kraus operators of the channel and of course one can always go from one parametrization to the other. 

%\red{Let us consider a straightforward way to obtain the Bloch parameters of a channel by using a fully known prepare-and-measure set-up which can differentiate between any two channels. Suppose we have a communication matrix $C'\in\mathcal{C}_{m,n}(\mathcal{H},\mathcal{K})$. By definition it can be implemented with some set of states $\{\varrho_j\}_{j\in[m]}\subset\mathcal{S}(\mathcal{H})$, some channel $\Phi: \lh \to \lhik$ and some measurement $M \in \mathcal{M}_n(\hik)$. Suppose that we know the states and the measurement and also that we can use the prepare-and-measure set-up to differentiate between any two channels. By Prop. \ref{prop:differentiability} the states $\{\varrho_j\}_{j\in[m]}$ and the measurement $M \in \mathcal{M}_n(\hik)$ are then informationally complete.}

\begin{proposition}\label{pr:blochparameters}
    To obtain the Bloch parametrization of a quantum channel $\Phi:\mathcal{L}(\mathcal{H})\to\mathcal{L}(\mathcal{K})$ using a communication matrix $C'_{\Phi(\varrho),M}\in\mathcal{C}_{m,n}(\mathcal{H,K})$ implemented with informationally complete states $\{\varrho_j\}_{j\in[m]}$ $\subset\mathcal{S}(\mathcal{H})$ and measurement $M \in \mathcal{M}_n(\hik)$, a fully known prepare-and-measure set-up is sufficient, but not necessary.
\end{proposition}

\begin{proof}
    In order to prove the sufficiency, let $\{\sigma_a\}_{a\in[d_{i}]}$ (respectively, $\{\sigma'_b\}_{b\in[d_{o}]}$) be the set of traceless self-adjoint operators that, along with $\sigma_0=\id_{d_i}$ (respectively, $\sigma'_0=\id_{d_o}$), constitute a generalized Bloch basis for $\mathcal{L}_s(\mathcal{H})$ (respectively, $\mathcal{L}_s(\mathcal{K})$). That is, $\tr{\sigma_a\sigma_{a'}}=d_{i}\delta_{aa'}$ and $\tr{\sigma'_b\sigma'_{b'}}=d_{o}\delta_{bb'}$ for all $a,a'\in[d_{i}^2-1]\cup\{0\}$ and $b,b'\in[d_{o}^2-1]\cup\{0\}$. From the informational completeness of the states and the measurement we get that
    \begin{align}
        \sigma_a &= \sum_{j\in[m]}\alpha^{(j)}_a\varrho_j \label{basis1}\\
        \sigma'_b &= \sum_{k\in[n]}\beta^{(k)}_bM_k,\label{basis2}
    \end{align}
    for some $\alpha^{(j)}_a,\beta^{(k)}_b\in\real$ for all $a \in [d_i]\cup \{0\}$ and $b \in [d_o]\cup \{0\}$ where the values of these coefficients can be solved if the states and the measurement are known. Now, recall that the parameters of the quantum channel $\Phi:\mathcal{L}(\mathcal{H})\to\mathcal{L}(\mathcal{K})$ in the Bloch representation are given by 
    \begin{equation}\label{parameters}
        \Phi_{ba} = \frac{1}{d_{o}}\tr{\Phi(\sigma_a)\sigma'_b}
    \end{equation}
    for all $a \in [d_i]\cup \{0\}$ and $b \in [d_o]\cup \{0\}$ so that for any $X = \sum_{c =0}^{d_i} \gamma_c \sigma_c  = \sum_{c =0}^{d_i} \frac{1}{d_i}  \tr{X \sigma_c} \sigma_c \in \lh$ we have that
    \begin{equation}
        \Phi(X) = \sum_{a=0}^{d_i} \sum_{b=0}^{d_o} \Phi_{ba} \gamma_a \sigma'_b = \frac{1}{d_i} \sum_{a=0}^{d_i} \sum_{b=0}^{d_o} \Phi_{ba}   \tr{X \sigma_a} \sigma'_b \, .
    \end{equation}
    We note that $\Phi_{0a}=\frac{d_{i}}{d_{o}}\delta_{0a}$ are fixed for all channels from $\mathcal{H}$ to $\mathcal{K}$. The remaining $d_{o}^2(d_{i}^2-1)$ parameters (when $b\neq0$) completely characterize the channel $\Phi$. Substituting Eqs.\eqref{basis1} and \eqref{basis2} in Eq.\eqref{parameters}, we get
    \begin{equation}
    \begin{split}
        \Phi_{ba} &= \frac{1}{d_{o}}\tr{\Phi(\sigma_a)\sigma'_b},\\
        &= \frac{1}{d_{o}}\tr{\Phi\left(\sum_{j\in[m]}\alpha^{(j)}_a\varrho_j\right)\sum_{k\in[n]}\beta^{(k)}_bM_k},\\
        &= \frac{1}{d_{o}}\left(\sum_{j\in[m]}\sum_{k\in[n]}\alpha^{(j)}_a\beta^{(k)}_b\tr{\Phi(\varrho_j)M_k}\right),\\
        &= \frac{1}{d_{o}}\left(\sum_{j\in[m]}\sum_{k\in[n]}\alpha^{(j)}_a\beta^{(k)}_bC'_{jk}\right), \label{eq:ch-tomo}
    \end{split}
    \end{equation}
    where $C'_{jk} = \tr{\Phi(\varrho_j)M_k}$ are the elements of the communication matrix $C'_{\Phi(\varrho),M}\in\mathcal{C}_{m,n}(\mathcal{H},\mathcal{K})$. Thus, all the parameters of the channel can be retrieved, which completely characterizes the channel.

    Now, we use the following example to show that complete knowledge of the prepare-and-measure set-up is not always necessary to parametrize a quantum channel.

    \begin{example}\label{ex1}
        Let us consider an informationally complete set of states $\{\varrho_j\}_{j\in[m]}\subset\mathcal{S}(\mathcal{H})$ and measurement $M \in \mathcal{M}_n(\hik)$. Suppose that the communication matrix $C' = (C'_{jk})_{j,k}$ obtained from $\{\varrho_j\}_j$ and $M$ in the presence of a channel $\Phi:\mathcal{L}(\mathcal{H})\to\mathcal{L}(\mathcal{K})$ is such that all rows of $C'$ are the same. That is, for every $k\in[n]$, we have $C'_{jk} = \tr{\Phi(\varrho_j)M_k} = \tr{\Phi(\varrho_{j'})M_k} = C'_{j'k}$ for all $j,j'\in[m]$. Since $M$ is informationally complete, this means that $\Phi(\varrho_j) = \Phi(\varrho_{j'}) = \xi\in\mathcal{S}(\mathcal{K})$ for all $j,j'\in[m]$, and since $\rspan{\{\varrho_j\}_j} = \mathcal{L}_s(\mathcal{H})$, this means that $\Phi$ is the single point contraction $\Phi(X) = \tr{X}\xi$ for all $X \in \lh$. It is not necessary to know the exact states $\{\varrho_j\}_j$ in this case to characterize the channel. Instead, it is sufficient that the measurement is known so that we can completely characterize $\xi$, thus characterizing the channel. Furthermore, if $\xi = \frac{\id_{d_o}}{d_o}$, it is not necessary to know exactly the measurement either. Instead, it is sufficient to know the trace of each effect $\tr{M_k}$ to characterize the channel (the completely depolarizing channel), since the elements of the matrix would be given exactly by $C'_{jk}=\frac{\tr{M_k}}{d_o}$ for all $j\in[m],k\in[n]$.
    \end{example}
\end{proof}

What we learn from above is that if we have a fully determined prepare-and-measure scenario consisting of states $\{\varrho_j\}$ and a measurement $M$, which can differentiate between any two channels, then given any channel $\Phi$, we can use the communication matrix $C_{\Phi(\varrho),M}$ to obtain all the channel parameters, and thus completely characterizing the channel. From Eq. \eqref{eq:ch-tomo} we see that in order for the described procedure to work, we indeed need to be able to differentiate between any two channels: if two channels give the same communication matrix, then Eq. \eqref{eq:ch-tomo} gives us the same parameters for the channels. Furthermore, even if complete information regarding the prepare-and-measure is not necessary, in order to be able to actually calculate the channel parameters for a given quantum channel, we need partial knowledge regarding the set-up (such as the trace of the measurement effects in Example \ref{ex1}). This highlights the destinction between tomography and differentiability: we need more information about the set-up in order to perform tomography.

\section{Semi-device-independent channel identification}\label{sec:sdi-identification}

In the previous section we have seen that a prepare-and-measure set-up can be used for channel tomography for any channel if and only if we can differentiate between any two channels and the full description of the set-up is known to us. On the other hand, differentation between any two channels is possible if and only if the set-up is informationally complete. Thus, in the case of channel tomography, full knowledge of the set-up is needed (from which informational completeness of the set-up can be checked) while for differentiation we can get away with only the knowledge about the informational completeness. The question know is can we somehow certify the information completeness of the set-up or even the full description of the set-up directly from the communication matrix, i.e., the statistics of the prepare-and-measure scenario. 

Thus, for our purposes let us consider a scenario where we have a prepare-and-measure set-up which we know is implemented with some $d$-dimensional quantum states and measurements. Suppose also that we can input any $d$-dimensional quantum channel into the set-up. However, suppose that we have restricted access to the set-up so that we only have access to the measurement statistics, i.e., the communication matrices before and after using the channels.

Let us start with the question about the informational completeness of the set-up. We will need some additional result about communication matrices. In \cite{HeinosaariKerppoLeppäjärvi2020} it was shown that the \emph{rank} of a communication matrix relates to the linear dimension of the theory in which it can be implemented in. In particular, in quantum theory for any communication matrix $C \in \mathcal{C}(\hi)$ on a $d$-dimensional Hilbert space $\hi$ we have that $\rank{C}\leq  \dim(\lsh)= d^2$. We can refine this connection further. For a communication matrix $C \in \mathcal{C}_{m,n}(\hi)$, let us denote the $k$th column vector and the $j$th row vector of $C$ by $\vec{C}_k$ and $\vec{C}^j$, respectively. Then $\rank{C} = \dim\left(\rspan{\{\vec{C}_k\}_{k=1}^n}\right)= \dim\left(\rspan{\{\vec{C}^j\}_{j=1}^m}\right)$, since, by definition, the rank of a matrix is equal to the number of linearly independent rows or columns. For a set of $m$ states $\{\varrho_j\}_{j=1}^m \subset \sh$ and a measurement $M \in \mathcal{M}_n(\hi)$ we denote $V_\varrho= \rspan{\{\varrho_j\}_{j=1}^m} \subseteq \lsh$ and $V_M= \rspan{\{M_k\}_{k=1}^n} \subseteq \lsh$. 

\begin{proposition}\label{prop:rank}
    Let $C \in \mathcal{C}_{m,n}(\hi)$ be a $m \times n$ communication matrix on a $d$-dimensional Hilbert space $\hi$. Then, 
    \begin{equation}
        \dim(V_\varrho \cap V_M )\leq \rank{C} \leq \min \left\lbrace \dim( V_\varrho), \dim( V_M)\right\rbrace
    \end{equation}
    for any implementation of $C = (C_{jk})_{j,k} = (\tr{M_k \varrho_j})_{j,k}$ consisting of an $n$-outcome measurement $M \in \mathcal{M}_n(\hi)$ and $m$ states $\{\varrho_j\}_{j=1}^m \subset \sh$.
\end{proposition}
\begin{proof}
    Let $M$ and $\{\varrho_j\}_{j=1}^m$ be some $n$-outcome measurement and $m$ states implementing $C$. Let us start with the second inequality. Suppose $d_{\min} := \min \left\lbrace \dim( V_\varrho), \dim( V_M)\right\rbrace = \dim\left( V_\varrho\right)$. Thus, there exists $d_{\min}$ states $\{\varrho_{j_x}\}_{x=1}^{d_{\min}}$ such that $\varrho_j = \sum_{x=1}^{d_{\min}} \alpha^{(j)}_x \varrho_{j_x}$ for some real numbers $\{\alpha^{(j)}_x\}_{x=1}^{d_{\min}} \subset \real$ for all $j \in [m]$. Now
    \begin{equation}
    \begin{split}
        \vec{C}^j &= (\tr{M_1 \varrho_j}, \ldots, \tr{M_n \varrho_j}) \\
        &= \left(\tr{M_1 \sum_{x=1}^{d_{\min}} \alpha^{(j)}_x \varrho_{j_x}}, \ldots, \tr{M_n \sum_{x=1}^{d_{\min}} \alpha^{(j)}_x \varrho_{j_x}}\right) \\
        &= \sum_{x=1}^{d_{\min}} \alpha^{(j)}_x (\tr{M_1 \varrho_{j_x}}, \ldots, \tr{M_n \varrho_{j_x}}) \\
        &= \sum_{x=1}^{d_{\min}} \alpha^{(j)}_x \vec{C}^{j_x}        
    \end{split}
    \end{equation}
    for all $j \in [m]$ so that $\rank{C} \leq d_{\min}$. Similarly, if we suppose that $\dim\left( \rspan{\{M_j\}_{k=1}^n}\right) = d_{\min}$, then we reach the same conclusion by noting that any linear combination of the effects of $M$ results in the linear combination of the columns of $C$. This proves the second inequality.

    In order to prove the first inequality, let us take a look on the properties of the linear function $g: V_\varrho \to V_M^*$ defined as $g(A) =f_A$ for all $A \in V_\varrho$, where $f_A$ is a linear function defined by the Hilbert-Schmidt inner product as $f_A(B) = \tr{A^* B} =\tr{AB}$ for all $B \in V_M$. By the rank-nullity theorem, we have that
    \begin{equation}
        \dim(V_\varrho) = \dim(\ker(g)) + \dim(\im(g)),
    \end{equation}
    where $\ker(g)$ and $\im(g)$ are the kernel and image of $g$. First we note that $g(A)=0$ if and only if $f_A(B) = \tr{AB} =0$ for all $B \in V_M$ which hold if and only if $A \in V_M^\perp = \{ X \in \lsh : \tr{BX}=0 \ \forall B \in V_M\}$. Thus, we have that $\ker(g) = V_\varrho \cap V_M^\perp$. On the other hand, we see that we can relate the rank of $C$ to $\dim(g)$. Namely, if $\rank{C}=r$, then there exists $r$ rows $\{\vec{C}^{(j_y)}\}_{y=1}^r$ such that $\vec{C}^{(j)} = \sum_{y=1}^r \beta^{(j)}_y \vec{C}^{(j_y)}$ for some real numbers $\{\beta^{(j)}_y\}_{y=1}^r \subset \real$ for all $j \in [m]$. We also note that now we can express the matrix elements of $C$ as $C_{jk} = \tr{\varrho_j M_k}= f_{\varrho_j}(M_k)$ for all $j \in [m]$ and $k \in [n]$. We now have that
    \begin{equation}
    \begin{split}
        (f_{\varrho_j}(M_1), \ldots, f_{\varrho_j}(M_n))&= \vec{C}^{(j)} = \sum_{y=1}^r \beta^{(j)}_y \vec{C}^{(j_y)} \\
        &= \left( \sum_{y=1}^r \beta^{(j)}_y f_{\varrho_{j_y}}(M_1), \ldots, \sum_{y=1}^r \beta^{(j)}_yf_{\varrho_{j_y}}(M_n) \right)
    \end{split}
    \end{equation}
    so that $f_{\varrho_j} = \sum_{y=1}^r \beta^{(j)}_y f_{\varrho_{j_y}}$ as linear functionals on $V_M$. Now if $A \in V_{\varrho}$ so that there exists real numbers $\{\gamma_j\}_{j=1}^m \subset \real$ such that $A= \sum_{j=1}^m \gamma_j \varrho_j$, we have that
    \begin{equation}
        f_A = \sum_{j=1}^m \gamma_j f_{\varrho_j} = \sum_{y=1}^r \left(\sum_{j=1}^m \gamma_j \beta^{(j)}_y \right) f_{\varrho_{j_y}}\, .
    \end{equation}
    Thus, we can conclude that $\dim(\im(g)) \leq r = \rank{C}$. Combining this with our earlier observations we have that
    \begin{equation}
        \rank{C} \geq \dim(V_\varrho) - \dim(V_\varrho \cap V_M^\perp) \,.
    \end{equation}
    To finalize the proof we note that since $V_M \oplus V_M^\perp = \lsh$ we have that $V_\varrho \cap V_M \oplus V_\varrho \cap V_M^\perp \subseteq V_\varrho$, which implies that $\dim(V_\varrho \cap V_M) + \dim(V_\varrho \cap V_M^\perp) \leq \dim(V_\varrho)$,
    and thus
    \begin{equation}
        \rank{C} \geq \dim(V_\varrho) - \dim(V_\varrho \cap V_M^\perp) \geq \dim(V_\varrho \cap V_M),
    \end{equation}
    which completes the proof.
\end{proof}

One can interpret the intuition behind the Proposition as follows: The rank of the matrix contains information about the degrees of freedom in which one can collect statistics. For example, consider a light source whose polarization is being measured. If one can measure polarization in all directions (H-V, A-D, and L-R), but the light source is always polarized in the horizontal or vertical axis, the statistics cannot characterize the noise in the diagonal or circular polarization. Similarly, if the source can produce all polarizations, but one can only measure in the diagonal axes, one cannot detect the noise in the horizontal-vertical or circular polarized light. On the other hand, one can also infer that the information we get from the statistics has to be at least as much as the overlap between the degrees of freedom of the preparation and measurement. Speaking in terms of light polarization, it can be seen that if we have light that is polarized in either the diagonal or the circular basis, and we can measure all the linear polarizations, then we will be able to characterize the noise in the diagonal axes of polarization at least.

It should be noted that in general we cannot hope to get rid of the inequalities, since there exist communication matrices where either one or both of the inequalities are strict. This can be shown with the following example.

\begin{example}
    Consider the qubit states $\varrho_1 = \kb{0}{0} = \frac{1}{2}(\id_2 + \sigma_z)$ and $\varrho_2 = \kb{+}{+} = \frac{1}{2}(\id_2 + \sigma_x)$, and the qubit measurement $M = \{M_\pm = \kb{\pm \iota}{\pm \iota} = \frac{1}{2}(\id_2 \pm \sigma_y)\}$. It can be seen that $d_{int} = \dim({V_\varrho\cap V_M}) = \dim({\{\vec{0}\}}) = 0$. Furthermore, since both $V_\varrho$ and $V_M$ are spaces each spanned by two linearly independent vectors in $\lsh$, $\dim(V_\varrho) = \dim(V_M) = 2 = d_{\min}$. However, the communication matrix obtained using these states $C$ is
    \begin{equation}
        C = \frac{1}{2}\begin{pmatrix}1 & 1\\1 & 1\end{pmatrix},
    \end{equation}
    which clearly shows that $\rank{C} = 1$, and thus $d_{int} < \rank{C} < d_{\min}$.
\end{example}

Using the previously obtained result, we can now see that we can deduce the informational completeness of the set-up directly from the communication matrix related to the set-up.

\begin{corollary}\label{corr:infocomplete}
     Let $C \in \mathcal{C}_{m,n}(\hi)$ be a $m \times n$ communication matrix on a $d$-dimensional Hilbert space $\hi$. Then, $\rank{C}=\dim(\lsh)=d^2$ if and only if any implementation of $C = (C_{jk})_{j,k} = (\tr{M_k \varrho_j})_{j,k}$ consists of an informationally complete measurement $M \in \mathcal{M}_n(\hi)$ and $m$ informationally complete states $\{\varrho_j\}_{j=1}^m \subset \sh$.
\end{corollary}
\begin{proof}
    If the states and the measurement are informationally complete, i.e., $V_\varrho = V_M = \lsh$, then from Prop. \ref{prop:rank} we see that then $\rank{C} = \dim(V_\varrho) = \dim(V_M) = \dim(\lsh)=d^2$. On the other hand, if $\rank{C}= d^2$, then from Prop. \ref{prop:rank} we have that $d_{\min} = \min(\dim(V_\varrho), \dim(V_M)) \geq d^2$ so that we must have $V_\varrho = V_M = \lsh$.
\end{proof}

This result shows us that given that we can trust the dimension of the prepare-and-measure set-up, we can deduce whether the set-up can be used for differentiating between any two quantum channels and channel tomography. However, as discussed before, for channel tomography we need to know also the description of the set-up in order to completely solve the channel parameters. While uniquely identifying the full description of the set-up just from the measurement statistics is a tall ask, it turns out that in some cases we can essentially self-test the set-up, i.e., identify the states and the measurement up to some symmetry \cite{Tavakolietal2018,FarkasKaniewski2019,Navascuesetal2023}. In particular, we will use another quantity of the communication matrix, namely its information storability, to draw conclusions about its implementation.

The \emph{information storability} of $C$, denoted by $\lmax{C}$, is defined as
    \begin{equation}
        \lmax{C} = \sum_{k} \max_j C_{jk} \, .
    \end{equation}
In \cite{HeinosaariKerppoLeppäjärvi2020,HeinosaariLeppäjärviPlavala2024} it was shown that the information storability relates to the amount of information that is possible to store in the theory in the way of encoding and decoding. Namely, if $C=C_{\varrho,M}$ for some states $\varrho=\{\varrho_j\}_{j=1}^m$ and and a measurement $M$, then $\lmax{C}$ tells us how much information we can decode by using the measurement $M$ when the message was encoded using the states $\varrho$, and $\lmax{C}/m$ is the average success probability for the minimum error discrimination task of discriminating the state $\varrho$ by using the measurement $M$. By optimizing $\lmax{C_{\varrho,M}}$ over all the states and the measurements we can recover the \emph{basic decoding theorem} \cite{SchumaherWestmorelandbook} which states that the error probability $P_E$ for decoding $m$ equally likely messages encoded in some $m$ states in a $d$-dimensional quantum theory, is lower-bounded by $P_E \geq 1- \frac{d}{m}$. In particular, in quantum theory for any communication matrix $C \in \mathcal{C}(\hi)$ on a $d$-dimensional Hilbert space $\hi$ we have that $\lmax{C} \leq \dim(\hi)=d$. We will see that when this inequality is saturated we can self-test the set-up.

\begin{proposition}\label{prop:implement}
     Let $C \in \mathcal{C}_{n,n}(\hi)$ be a $n \times n$ communication matrix on a $d$-dimensional Hilbert space $\hi$ such that it has no zero columns.  We have $\lmax{C}=d$ if and only if it can be implemented by a measurement $M \in\mathcal{M}_n(\hi)$ consisting of rank-1 effects $\{M_k\}_{k=1}^n$ and the implementing states $\{\varrho_j\}_{j=1}^n$ are precisely the pure eigenstates of the effects. Furthermore, in this case, for any other implementing measurement $M'$ and states $\{\varrho'_j\}_{j=1}^n$, there exists an orthogonal transformation $h$ belonging to a subgroup of $O(d)$ isomorphic to $O(d^2-1)$ such that
     \begin{align}
         M'_k &= h(M_k) \\
         \varrho'_j &= h(\varrho_j) 
     \end{align}
     for all $j,k \in [n]$.  For qubits, $h$ becomes equivalent to a unitary or anti-unitary transformation.
\end{proposition}
\begin{proof}
    The first part was shown in \cite[Prop. 1]{HeinosaariKerppoLeppäjärvi2020} (see also \cite{HeinosaariLeppäjärviPlavala2024}). For completeness, we repeat the argument here. Let $M$ and $\{\varrho_j\}_{j=1}^n$ be some implementation of $C$. If we denote the largest eigenvalue of $M_k$ by $\lambda_k$, we see that
    \begin{align}
        \lmax{C} &= \sum_{k=1}^n \max_{j \in [n]} C_{jk} = \sum_{k=1}^n \max_{j \in [n]}\tr{M_k \varrho_j} \label{eq:lmax-1} \\
        &\leq \sum_{k=1}^n \lambda_k \leq \sum_{k=1}^n \tr{M_k} = \tr{\id_d} = d,\label{eq:lmax-2}
    \end{align}
    where we have used the fact that $\tr{M_k \varrho_j} \leq \lambda_k$ for all $k \in [n]$.
    The equality can only hold when each $M_k$ is rank-1 operator and when the states $\{\varrho_j\}_{j=1}^n$ are exactly the eigenstates of each $M_k$.
    
    Thus, there exists some $n$ unit vectors $\{\varphi_j\}_{j=1}^n \in \hi$ and some positive real numbers $\{\alpha_j \}_{j=1}^n \subset [0,1]$ such that (up to possible bijective relabeling of the states or effects) we have $M_k = \alpha_k \kb{\varphi_k}{\varphi_k}$ and $\varrho_j = \kb{\varphi_j}{\varphi_j}$ for all $j,k \in [n]$. Thus, we have that $C_{jk}= \alpha_k |\ip{\varphi_k}{\varphi_j}|^2$ for all $j,k \in [n]$. In particular, $C_{jj} = \alpha_j$ for all $j \in [n]$ and thus $\sum_{j=1}^n \alpha_j = d$. We note that this can hold only if $n\geq d$.

    Suppose that there exists another implementation with some rank-1 measurement $M'$ and corresponding eigenstates $\{\varrho'_j\}_{j=1}^n$. We then have that $M'_k = \alpha'_k \kb{\varphi'_k}{\varphi'_k}$ and $\varrho'_j = \kb{\varphi'_j}{\varphi'_j}$ for all $j,k \in [n]$ for some $n$ unit vectors $\{\varphi'_j\}_{j=1}^n \in \hi$ and some positive real numbers $\{\alpha'_j \}_{j=1}^n \subset [0,1]$. Now thus $C_{jk} = \alpha'_k |\ip{\varphi'_k}{\varphi'_j}|^2$  for all $j,k \in [n]$. Again we have that $C_{jj} = \alpha'_j$ for all $j \in [n]$ so that we can conclude that $\alpha_j = \alpha'_j$ for all $j \in [n]$. Then comparing the rest of the elements of $C$ we see that we must have 
    \begin{equation}
        |\ip{\varphi_k}{\varphi_j}| = |\ip{\varphi'_k}{\varphi'_j}|
    \end{equation}
    for all $j,k \in [n]$. 

    Let us denote $V = V_\varrho = V_M= \rspan{\{\kb{\varphi_j}{\varphi_j}\}_{j=1}^n} \subseteq \lsh$ and similarly $V' = V_{\varrho'} = V_{M'} =\rspan{\{\kb{\varphi'_j}{\varphi'_j}\}_{j=1}^n} \subseteq \lsh$ for the two vector subspaces of $\lsh$ spanned by the states (or equivalently the effects). Let us consider a linear map $f: V \to V'$ defined as $f(\kb{\varphi_j}{\varphi_j}) = \kb{\varphi'_j}{\varphi'_j}$ for all $j \in [n]$ and which is extended to $V$ by linearity so that $f\left( \sum_j \beta_j \kb{\varphi_j}{\varphi_j}\right) = \sum_j \beta_j \kb{\varphi'_j}{\varphi'_j}$ for all $\{\beta_j\}_j \subset \real$. It is clear that $f$ is unital, as
    \begin{equation}
    \begin{split}
            f\left(\id_d\right) &= f\left(\sum_k M_k\right) \\
            &= f\left(\sum_k \alpha_k \kb{\varphi_k}{\varphi_k}\right) \\
            &= \sum_k \alpha_k f\left(\kb{\varphi_k}{\varphi_k}\right) \\
            &= \sum_k \alpha'_k \kb{\varphi'_k}{\varphi'_k} \\
            &= \sum_k M'_k = \id_d.
    \end{split}
    \end{equation}

    Let $\tilde{d} = \dim\left(V\right)$. With some relabeling of the states and the effects we can assume that $\rspan{\{\kb{\varphi_j}{\varphi_j}\}_{j=1}^{\tilde{d}}}= V$ so that $\{\kb{\varphi_j}{\varphi_j}\}_{j=1}^{\tilde{d}}$ forms a basis of $V$. Thus, for all $\kb{\varphi_k}{\varphi_k}$ for all $k \in [n]$ there exists some real numbers $\{\beta^{(k)}_j\}_{j=1}^{\tilde{d}} \subset \real$ such that $\kb{\varphi_k}{\varphi_k} = \sum_{j=1}^{\tilde{d}}\beta^{(k)}_j \kb{\varphi_j}{\varphi_j}$. Now we see that
    \begin{equation}
        \kb{\varphi'_k}{\varphi'_k} = f(\kb{\varphi_k}{\varphi_k}) = \sum_{j=1}^{\tilde{d}} \beta^{(k)}_j \kb{\varphi'_j}{\varphi'_j}
    \end{equation}
    for all $k \in [n]$. This shows that the states $\{\kb{\varphi'_j}{\varphi'_j}\}_{j=1}^{\tilde{d}}$ spans $V'$. On the other hand, since the two implementations realize the same communication matrix $C$, and since $V = V_\varrho = V_M$ and $V' = V_{\varrho'} = V_{M'}$, from Prop. \ref{prop:rank} we can conclude that $\dim(V') = \rank{C}= \dim(V) = \tilde{d}$. Thus, the states $\{\kb{\varphi'_j}{\varphi'_j}\}_{j=1}^{\tilde{d}}$ form a basis of $V'$. Since $f$ maps the basis of $V$ to the basis of $V'$ we can conclude that $f$ is a bijection between $V$ and $V'$. Furthermore, if $A = \sum_{j=1}^n \alpha_j \kb{\varphi_j}{\varphi_j} \in V$ and $B = \sum_{k=1}^n \beta_k \kb{\varphi_k}{\varphi_k} \in V$, then
    
    \begin{equation}
    \begin{split}
        \ip{A}{B}_{HS} &= \tr{A^*B} = \tr{AB} \\
        &= \sum_{j,k=1}^n \alpha_j \beta_k |\ip{\varphi_j}{\varphi_k}|^2 \\
        &= \sum_{j,k=1}^n \alpha_j \beta_k |\ip{\varphi'_j}{\varphi'_k}|^2\\
        &= \tr{f(A)f(B)} = \tr{f(A)^*f(B)} \\
        &= \ip{f(A)}{f(B)}_{HS}        
    \end{split}
    \end{equation}
    so that $f$ is actually an isometry.
    
    Let us now decompose $\lsh$ into direct sums $\lsh = V \oplus V^\perp = V' \oplus V'^\perp$ with respect to the Hilbert-Schmidt inner product. Since $\dim(V) = \dim(V')$, we must have that $\dim(V^\perp) = \dim(V^\perp)$. Thus, there exists an isomorphism $g: V^\perp \to V'^\perp$ that preserves the Hilbert-Schmidt inner product. We can now use $f$ and $g$ to build a linear map $h: \lsh \to \lsh$ defined as
    \begin{equation}
        h(A) = f(A_V)+g(A_{V^\perp})
    \end{equation}
    for all $A = A_V + A_{V^\perp} \in \lsh$, where $A_V \in V$ and $A_{V^\perp} \in V^\perp$. We note that $h$ is clearly a bijection since $f$ maps a basis of $V$ to a basis of $V'$ and $g$ maps a basis of $V^\perp$ to a basis of $V'^\perp$ and since $\lsh = V \oplus V^\perp = V' \oplus V'^\perp$ we see that altogether $h$ maps a basis of $\lsh$ to a basis of $\lsh$. Furthermore, we have that for all $A = A_V + A_{V^\perp} \in \lsh$ and $B = B_V + B_{V^\perp} \in \lsh$:
    \begin{equation}
    \begin{split}
        \ip{h(A)}{h(B)}_{HS} =& \ip{h(A_V + A_{V^\perp})}{h(B_V + B_{V^\perp})}_{HS} \\
        =& \ip{f(A_V)}{f(B_V)}_{HS}+\ip{f(A_V)}{g(B_{V^\perp})}_{HS} \\
        &+\ip{g(A_{V^\perp})}{f(B_V)}_{HS}+\ip{g(A_{V^\perp})}{g(B_{V^\perp})}_{HS} \\
        =& \ip{A_V}{B_V}_{HS}+0+0+\ip{A_{V^\perp}}{B_{V^\perp}}_{HS} \\
        =& \ip{A_V}{B_V}_{HS}+\ip{A_V}{B_{V^\perp}}_{HS}\\
        &+\ip{A_{V^\perp}}{B_V}_{HS}+\ip{A_{V^\perp}}{B_{V^\perp}}_{HS} \\
        =& \ip{A}{B}_{HS} \, .
    \end{split}
    \end{equation}
    Thus, $h$ is an isometry on $\lsh$. Since $h$ is a bijective isometry, it is an orthogonal transformation on $\lsh$, that is, $h\in O(d^2)$. Furthermore, the unitality of $h$ makes it an element of a subgroup of $O(d^2)$ that is isomorphic to $O(d^2-1)$.
    In particular, we have that 
    \begin{align}
        \varrho'_j &= \kb{\varphi'_j}{\varphi'_j} = h(\kb{\varphi_j}{\varphi_j}) = h(\varrho_j), \\
        M'_k &= \alpha_k \kb{\varphi'_k}{\varphi'_k} = \alpha_k h(\kb{\varphi_k}{\varphi_k}) = h(\alpha_k \kb{\varphi_k}{\varphi_k}) = h(M_k)
    \end{align}
    for all $j, k \in [n]$ as claimed. In the case of qubits ($d = 2$), we thus have that $h\in O(3)$. Since $O(3)$ is isomorphic to the projective unitary-antiunitary group $PUA(2)$, the map $h$ corresponds to a unitary or antiunitary transformation on qubits. 
\end{proof}

Thus, we see that the information storability of a communication matrix can be used to self-test some class of prepare-and-measure set-ups. As we saw, this type of set-up includes all set-ups where we use some rank-1 measurements together with their eigenstates. From the practical point of view, this covers all the standard experimental cases where the prepare-and-measure scenario is implemented by pure states related to some orthogonal bases of $\hi$ and where we measure in those bases. Namely, let us suppose that we have $N$ orthonormal bases $\{\psi^{(i)}_j\}_{j\in[d], i \in [N]}$ for $\hi$. We can define states $\varrho^{(i)}_j = \kb{\psi^{(i)}_j}{\psi^{(i)}_j}$ and $N$ basis measurements $M^{(i)}$ each with $d$-outcomes as $M^{(i)}_j = \kb{\psi^{(i)}_j}{\psi^{(i)}_j}$ for all $i \in [N]$ and $j \in [d]$. Let us define a $N \times d$-outcome measurement $M$ as $M_{(i,j)}=\frac{1}{N} M^{(i)}_j $  for all $i \in [N]$ and $j \in [d]$. Operationally this just corresponds to uniform mixing of the basis measurements such that we keep track of which measurement we use in each round of mixing. The resulting communication matrix, implemented only with states corresponding to orthonormal bases and measurements that correspond to measurements in those bases, will have $\lmax{C}=d$. For example, in qubit systems, the standard channel tomography using $x,y,z$ spin directions can be implemented in this way. According to our result we can self-test these standard set-ups up to a unitary or antiunitary freedom for qubits, and an orthogonal transformation for general $d$-dimensional systems.

\begin{remark}[Robustness]
    We can show that the self-testing strategy based on maximizing the information storability is robust to noise. To see this, let $C \in \mathcal{C}_{n,n}(\hi)$ be a $n \times n$ communication matrix on a $d$-dimensional Hilbert space $\hi$ such that it has no zero columns. Let now $\lmax{C}=d-\epsilon<d$ for some $\epsilon >0$. Let $\{\varrho_j\}_{j \in [n]}$ and $M \in \mathcal{M}_n(\hi)$ be some implementation of $C$ so that $C = C_{\varrho, M}$. Let the spectral decomposition of each effect be $M_k = \sum_{i=1}^d \lambda^{(k)}_i \kb{\varphi^{(k)}_i}{\varphi^{(k)}_i}$ so that $\{\varphi^{(k)}_i\}_{i=1}^d$ is an orthonormal basis for $\hi$ and $\lambda^{(k)}_1 \geq \ldots \geq \lambda^{(k)}_d$ are the eigenvalues of $M_k$ for all $k \in [n]$. For each $ k \in [n]$ let $\varrho_{j_k} \in \{\varrho_j\}_{j=1}^n$ be some state such that $\tr{M_k \varrho_{j_k}} = \max_j \tr{M_k \varrho_j}$. Now we have that
    \begin{align}\label{eq:rob0}
        \epsilon = d - \lmax{C} = \sum_{k=1}^n \left( \tr{M_k} - \tr{M_k \varrho_{j_k}} \right) = \sum_{k=1}^n \sum_{i=1}^d \left( \lambda^{(k)}_i -  \lambda^{(k)}_i \ip{\varphi^{(k)}_i}{ \varrho_{j_k} \varphi^{(k)}_i} \right) 
    \end{align}
    from which we see that 
    \begin{align} 
        \lambda^{(k)}_i\left(1-\ip{\varphi^{(k)}_i}{ \varrho_{j_k} \varphi^{(k)}_i}\right) \leq \sum_{i=1}^d \left(\lambda^{(k)}_i -  \lambda^{(k)}_i \ip{\varphi^{(k)}_i}{ \varrho_{j_k} \varphi^{(k)}_i}\right) < \epsilon 
    \end{align}
    since for every $k \in [n]$ and $i \in [d]$ every such term is positive. On the other hand, since $\lambda^{(k)}_1 \leq \sum_{i=1}^d\lambda^{(k)}_i$ and $\tr{\varrho_{j_k}}=1$, we have that
    \begin{align}
        \sum_{i=1}^d \left(\lambda^{(k)}_1-  \lambda^{(k)}_i\right) \ip{\varphi^{(k)}_i}{ \varrho_{j_k} \varphi^{(k)}_i}   \leq  \sum_{i=1}^d \left(\lambda^{(k)}_i-  \lambda^{(k)}_i \ip{\varphi^{(k)}_i}{ \varrho_{j_k} \varphi^{(k)}_i}\right) < \epsilon 
    \end{align}
    from which we get that
    \begin{align}
        \left(\lambda^{(k)}_1-  \lambda^{(k)}_i\right) \ip{\varphi^{(k)}_i}{ \varrho_{j_k} \varphi^{(k)}_i} < \epsilon
    \end{align}
    for all $k \in [n]$ and $i \in [d]$. Rewriting Eqs. \eqref{eq:rob1} and \eqref{eq:rob2} in terms of the fidelity $F$ of quantum states, defined as $F(\xi,\varrho) = \left(\tr{\left(\xi^{\frac{1}{2}} \varrho \xi^{\frac{1}{2}}\right)^{\frac{1}{2}}}\right)^2 $ for $\xi, \varrho \in \sh$, we have that
    \begin{align}
         \lambda^{(k)}_i\left[1-F\left(\kb{\varphi^{(k)}_i}{\varphi^{(k)}_i},\varrho_{j_k}\right) \right] &< \epsilon,  \label{eq:rob1} \\
         \left(\lambda^{(k)}_1-  \lambda^{(k)}_i\right) F\left(\kb{\varphi^{(k)}_i}{\varphi^{(k)}_i},\varrho_{j_k}\right) &< \epsilon \label{eq:rob2}
    \end{align}
    for all $k \in [n]$ and $i \in [d]$.
    
    Now we see that if $\epsilon$ is small, then according to Eq. \eqref{eq:rob1} either each eigenvalue $\lambda^{(k)}_i$ is small or $\varrho_{j_k}$ is very close to $\kb{\varphi^{(k)}_i}{\varphi^{(k)}_i}$. On the other hand, according to Eq. \eqref{eq:rob2} either each eigenvalue $\lambda^{(k)}_i$ is close to its maximum eigenvalue $\lambda^{(k)}_1$ or $\varrho_{j_k}$ is almost orthogonal to the eigenstate $\kb{\varphi^{(k)}_i}{\varphi^{(k)}_i}$. Combinining these observations together we have that for a small enough $\epsilon$ we have that a) either an eigenvalue $\lambda^{(k)}_i$ of $M_k$ is very small, in which case the corresponding eigenstate $\kb{\varphi^{(k)}_i}{\varphi^{(k)}_i}$ is almost orthogonal to the maximizing state $\varrho_{j_k}$, or b) the maximizing state $\varrho_{j_k}$ is close to the eigenstate $\kb{\varphi^{(k)}_i}{\varphi^{(k)}_i}$ corresponding to the eigenvalue $\lambda^{(k)}_i$ in which case $\lambda^{(k)}_i$ is close to the maximum eigenvalue $\lambda^{(k)}_1$. Since all the eigenstates are orthogonal, there can be only one eigenstate which satisfies condition b). In fact, from Eq. \eqref{eq:rob0} we can confirm that 
    \begin{align}
        \sum_{k=1}^n \sum_{i\neq 1} \lambda^{(k)}_i \leq \epsilon,
    \end{align}
    because $\tr{M_k \varrho_{j_k}}\leq \lambda^{(k)}_1$ for all $k \in [n]$. Thus, even in the presence of some small noise, the implementation must be close to the one given in Prop. \ref{prop:implement}.
\end{remark}

Getting back to identifying channels, we see that we can use the previous result to characterize a given channel up to an orthogonal transformation. Since the self-testing does not allow us to fix the frame of reference, we cannot use it to perform complete channel tomography. However it allows us to self-test the whole prepare-transform-measure scenario which allows us to characterize the channel up to an orthogonal transformation.

\begin{corollary}
    Let $C \in \mathcal{C}(\hi)$ be communication matrix on a $d$-dimensional Hilbert space $\hi$. If $\rank{C} = d^2$, the measure-and-prepare scenario implementing $C$ can be used to differentiate between any two quantum channels on $\hi$. Furthermore, if $\lmax{C}=d$, then we can characterize the channel up to an orthogonal transformation.
\end{corollary}

The consequences of the corollary can be demonstrated with an example as follows:

\begin{example}\label{eg:D4}
    Let $d=2$ and suppose that the communication matrix is a distinguishability matrix \cite{HeinosaariKerppoLeppäjärvi2020} given as follows:
    \begin{equation}\label{D4}
        C = D_{4,\frac{1}{2}} = \frac{1}{6}\begin{pmatrix}3 & 1 & 1 & 1\\1 & 3 & 1 & 1\\1 & 1 & 3 & 1\\1 & 1 & 1 & 3\end{pmatrix}.
    \end{equation}
    It can be seen that $\rank{C} = d^2$ and $\lmax{C}=d$. According to Prop. \ref{prop:implement}, $C$ can be implemented by a rank-1 four outcome POVM, and the eigenstates of the rank-1 effects. Furthermore, this implementation is unique up to unitaries. It is already known that $C$ can be obtained using the states:
    \begin{align}
    \varrho_0 &= \frac{1}{2}\left(\id_2 + \frac{1}{\sqrt{3}}(\sigma_x + \sigma_y + \sigma_z)\right);\label{eq:SICState0}\\
    \varrho_j &= \frac{1}{2}\left(\id_2 + \frac{1}{\sqrt{3}}\sum_{a=1}^3(-1)^{(1+\delta_{ja})}\sigma_a\right),\qquad j\in\{1, 2, 3\},\label{eq:SICStates}
    \end{align}
    and the measurement given by the effects $M_k = \frac{1}{2}\varrho_k,$ $k\in\{0,1,2,3\}$. Recall that for qubits, the orthogonal transformations are equivalent to a unitary or antiunitary transformation. Since any unitary or antiunitary transformation on the states (and measurements) can be equivalently considered as a transformation of the operator basis $\{\id_2,\sigma_a\}_{a\in\{x,y,z\}}$ into $\Omega = \{\id_2,\sigma'_a = U\sigma_aU^*\}_{a\in\{x,y,z\}}$ for some unitary or antiunitary $U$, any implementation of $C$ can be written using states of the form above for some operator basis $\Omega$ according to Prop. \ref{prop:implement}. The channel parameters $\Phi$ in the operator basis $\Omega$ can thus be retrieved by following the steps in Eq.\eqref{eq:ch-tomo}.
\end{example}

What happens if we do not have an informationally complete set of states and/or measurements? Let us take a look at another example:

\begin{example}\label{eg:D3}
    Let $d=2$ and suppose that the communication matrix is yet another distinguishability matrix given as follows:
    \begin{equation}\label{D3}
        C = D_{3,\frac{1}{3}} = \frac{1}{6}\begin{pmatrix}4 & 1 & 1\\1 & 4 & 1\\1 & 1 & 4\end{pmatrix}.
    \end{equation}
    Evidently, $\lmax{C} = 2$, and it has been shown in \cite{HeinosaariKerppoLeppäjärvi2020} that this matrix can be implemented using the following states
    \begin{align}
    \varrho_1 &= \frac{1}{2}(\id_2 + \sigma_z);\label{eq:trine1}\\
    \varrho_2 &= \frac{1}{2}(\id_2 + \frac{\sqrt{3}}{2}\sigma_x - \frac{1}{2}\sigma_z);\label{eq:trine2}\\
    \varrho_3 &= \frac{1}{2}(\id_2 - \frac{\sqrt{3}}{2}\sigma_x - \frac{1}{2}\sigma_z),\label{eq:trine3}
    \end{align}
    and the measurement $M = \{M_k = \frac{2}{3}\varrho_k\}$. It can be seen that $\id_2,\sigma_x,\sigma_z\in\rspan{\{\varrho_j\}} = \rspan{\{M_k\}}$. So, using the method described in Eq.\eqref{eq:ch-tomo}, we can retrieve any parameter $\Phi_{ba}$ of the channel $\Phi$, where $b,a\in\{0,1,3\}$, in the operator basis $\{\id_2,\sigma_x,\sigma_y,\sigma_z\}$, but not the parameters where $b=2$ or $a=2$.
    
    However, the matrix $D_{3,\frac{1}{3}}$ can still be useful in characterizing certain qubit channels. For example, if $D_{3,\frac{1}{3}}$ is the communication matrix after introducing the channel, we can conclude that the states after passing through the channel are unitarily equivalent to the states above, which are points on a great circle on the Bloch sphere and forms an equilateral triangle. Since quantum channels are contractive \cite{heinosaari2011mathematical}, the trace distances are non-increasing, that is, for all $j,j'\in\{1,2,3\}$
    \begin{equation}
        ||\Phi(\varrho_{j})-\Phi(\varrho_{j'})||_{\text{tr}} \leq ||\varrho_{j}-\varrho_{j'}||_{\text{tr}},
    \end{equation}
    which is only possible if $\{\varrho_j\}$ also form an equilateral triangle on a great circle in the Bloch sphere. This also means that $\{\varrho_j\}$ span that great circle and there is only one additional dimension where the channel can contract the states. However, according to the \emph{no-pancake theorem}\cite{Braun_2014,Blumeetal2010}, a map that is strictly contractive in only one direction of the Bloch sphere is not completely positive, and thus, is not a quantum channel. Therefore, the only possible options of the channel are to do nothing or flip the poles of the Bloch sphere orthogonal to the great circle spanned by $\{\varrho_j\}$. In other words, the channel is a reversible transformation from one operator basis to another, and is thus given by a unitary or antiunitary, or is the identity channel up to unitary or antiunitary transformation.
\end{example}

\section{Additional information or resources}\label{sec:additional}

Let us next consider scenarios where we can access some additional information about the channel or some additional resources at our disposal. This is information can be given to us a form of prior information or we can consider scenarios where we can try to get this additional information ourselves. In particular we will look into cases when we can use the channel sequentially multiple times and in this way get information about the powers of the channel, or when it is given to us that the channel is unital. 

\subsection{Multiple uses of the channel}

We start with the scenario when we have multiple uses of the channel sequentially for example by always just taking the transformed state after each use if the channel and sending it through the channel again. Naturally in this case we have to have $\hik=\hi$. Alternatively such a scenario can arise when we have a channel preparator which prepares us identical channels and we just apply them one after another. One could think that in the case where one gets all the prepare-and-measure statistics after one to $\lambda$ uses of the channel, that it would contain more information about the channel so that maybe the informational completeness is no longer needed. However, even in this case we see that informational completeness of the prepare-and-measure set-up is needed for differentiating any two channels. 

\begin{corollary}
    A set of quantum states $\{\varrho_j\}_{j\in[m]}\subset\mathcal{S}(\mathcal{H})$ and a quantum measurement $M \in \mathcal{M}_n(\hi)$ can differentiate any two quantum channels $\Phi_1,\Phi_2:\mathcal{L}(\mathcal{H})\to\mathcal{L}(\mathcal{H})$ after up to $\lambda\in\nat$ uses of the channel if only if the states and the measurement are informationally complete.
\end{corollary}
\begin{proof}
    Clearly if the states and the measurements are informationally complete, then from Prop. \ref{prop:differentiability} we see that they can differentiate between any two channels already if $\lambda=1$. The rest of proof is a straightforward extension of the forward direction of the proof of Prop. \ref{prop:differentiability}. We have $\mathcal{H}=\mathcal{K}$, so $d_{i} = d_{o} = d$, and $\id_{d_i} = \id_{d_o}$.\\
    \textit{Case 1:} We choose the states $\xi_1,\xi_2\in V_\varrho = \rspan{\{\varrho_j\}_j}$. It is straightforward to see from \textit{Case 1} of Prop. \ref{prop:differentiability} that $\Phi_{1}^{\lambda}(X) = \Phi_{1}(X)$ and $\Phi_{2}^{\lambda}(X) = \Phi_{1}(X)$. The rest follows from the  proof of Prop. \ref{prop:differentiability}.\\
    \textit{Case 2:} It is easy to see from \textit{Case 2} of Prop. \ref{prop:differentiability} that $\Phi_{1}^{\lambda}(X) = \Phi_{1}(X) = \tr{X}\frac{\id_{d}}{d}$ for all $X\in\mathcal{L}(\mathcal{H})$. Furthermore, $\tr{\Phi_2(X)M_k} = \frac{\tr{M_k}}{d}$ is also true for all $k$ and for all $X\in\{\Phi^\lambda_2(\varrho_j)\},\lambda\in\nat$. The rest follows from the proof of Prop. \ref{prop:differentiability}.
\end{proof}

\begin{remark}
    In the previous Corollary we assumed that we have access to all the communication matrices after one to $\lambda$ uses of the channel. We note that informational completeness of the states $\{\varrho_j\}_j$ and the measurement $M$ is not sufficient to differentiate between any two quantum channels $\Phi_1,\Phi_2:\mathcal{L}(\mathcal{H})\to\mathcal{L}(\mathcal{H})$ if we have access to only exactly $\lambda\in\nat$ uses of the channel, since there exist multiple channels such that $\Phi^\lambda = id_\hi$, where $id_\hi$ is the identity channel on $\hi$ (eg: Pauli gates when $\mathcal{H} = \complex^2$ and $\lambda=2$).
\end{remark}

\subsection{Unital channels} In the case that we have some prior information about the channel, we do not necessarily need the full informational completeness of the prepare-and-measure set up. As an instance of such prior information, we can see what happens when we restrict to \emph{unital} channels, i.e., channels $\Phi: \lh \to \lh$ such that $\Phi(\id_d)=\id_d$.

\begin{proposition}\label{prop:unitaldiff}
    Let $\hi$ be a $d$-dimensional Hilbert space. 
    A set of states $\{\varrho_j\}_{j\in[m]}\subset\mathcal{S}(\mathcal{H})$ and a measurement $M \in \mathcal{M}_n(\hi)$ can differentiate between any two unital quantum channels $\Phi_1,\Phi_2:\mathcal{L}(\mathcal{H})\to\mathcal{L}(\mathcal{H})$ if and only if $M$ is informationally complete and there exists a set $S$ of $d^2-1$ linearly independent states in $\{\varrho_j\}_{j\in[m]}$ such that $\id_d \notin \rspan{S}$.
\end{proposition}
\begin{proof}
    It is known that for a set $S$ of linearly independent states we have that $\id_d \in S$ if and only if the Bloch vectors of the states in $S$ are linearly dependent.
    
    ($\impliedby$) Let us suppose we have an informationally complete measurement $M \in \mathcal{M}_n(\hi)$ and a set of states $\{\varrho_j\}_{j\in[m]}\subset\mathcal{S}(\mathcal{H}),$ $m\geq d^2-1$ such that there exists a set $S\subset\{\varrho_j\}$ of $d^2-1$ linearly independent states satisfying $\id_d \notin \rspan{S}$. Denote $S' = \{\id_d\}\cup S$, $V_S = \rspan{S}$, and $V_{S'} = \rspan{S'}$. Since $S'$ contains only linearly independent elements, we have that $\dim(V_{S'}) = d^2$, which means that $V_{S'} = \lsh$. Furthermore, $\cspan{S'} = \lh$. In other words, $S'$ is informationally complete. Now, suppose that we have two unital channels $\Phi_1,\Phi_2:\mathcal{L}(\mathcal{H})\to\mathcal{L}(\mathcal{H})$ that cannot be differentiated by the states in $S$ and the measurement $M$. That is, $\tr{\Phi_1(\varrho_{j_x})M_k} = \tr{\Phi_2(\varrho_{j_x})M_k}$ for all $\varrho_{j_x}\in S$ and for all $k$. Since $M$ is informationally complete, we have that $\Phi_1(\varrho_{j_x}) = \Phi_2(\varrho_{j_x})$ for all $\varrho_{j_x}\in S$ and from the unitality of the channel, we have that $\Phi_1(\id_d) = \Phi_2(\id_d) = \id_d$. Which means that $\Phi_1(X) = \Phi_2(X)$ for all $X\in V_{S'} = \lsh$, or $\Phi_1(X) = \Phi_2(X)$ for all $X\in \cspan{S'} = \lh$, that is, $\Phi_1 = \Phi_2$. Thus, the set of states $\{\varrho_j\}$ containing linearly independent states $S$ such that $\id_d \notin S$, and an informationally complete measurement $M$ can differentiate between any two unital channels on $\hi$.
     
    ($\implies$) The channels used as examples in \textit{Case 2} of the proof for Prop. \ref{prop:differentiability} are unital when $\mathcal{H}=\mathcal{K}$. This means that even if we restrict ourselves to the set of unital quantum channels instead of the set of all quantum channels, the measurement has to be informationally complete to perfectly distinguish any two arbitrary quantum channels.
    
    Let us consider an informationally complete quantum measurement $M \in \mathcal{M}_n(\hi)$ and a set of $d$-dimensional quantum states $\{\varrho_j\}_{j\in[m]}\subset\mathcal{S}(\mathcal{H})$, $\dim(\mathcal{H})=d$. Let $S$ be any subset of $\{\varrho_j\}$ with $d^2-1$ elements, denote $V_S = \rspan{S}$. Consider that for all such $S$, at least one of the following is true:
    \begin{enumerate}
        \item $S$ is linearly dependent.
        \item $\id_d\in V_{S}$.
    \end{enumerate}
        
    \textit{Case 1:} Suppose that all such $S$ are linearly dependent. Then, $\dim(V_S) \leq d^2-2$ for all $S$, and therefore, $\dim(V_\varrho) \leq d^2-2$. This means that there exist at least two normalized orthogonal traceless self-adjoint operators that are in the orthogonal subspace $V_\varrho^\perp$.

    \textit{Case 2:} Let us now consider a linearly independent $S$ such that $\id_d\in V_S$. Suppose that $S'\subset\{\varrho_j\}$ is another such subsets distinct from $S$. Since $S$ and $S'$ are linearly independent sets, we have that $\dim(V_S) = \dim(V_{S'}) = d^2-1$. Suppose that $V_S \neq V_{S'}$. This means that there exists $\varrho_{j}\in S'$ for some $j$ such that $\varrho_{j}\notin V_S$. Let $\varrho_{j'}\in S'$ be one such state. Therefore, the set $\{\varrho_{j'}\}\cup S\subset\{\varrho_j\}_{j\in[n]}$ has $d^2$ linearly independent states, and is a basis for $\lsh$. Since it is already given that $\id_d\in V_S$, we can write that
    \begin{equation}
        \id_d = \alpha_{j_1}\varrho_{j_1} + \sum_{x=2}^{d^2-1}\alpha_{j_x}\varrho_{j_x},
    \end{equation}
    where $\varrho_{j_x}\in S$ for all $x\in[d^2-1]$. Without loss of generality, let $\alpha_1\neq0$. Now, define a set $S_0 = \{\varrho_{j'}\}\cup (S\backslash\{\varrho_1\})$. Clearly, $S_0$ is a subset of $\{\varrho_j\}$ with $d^2-1$ linearly independent elements. Furthermore, $\id_d\notin V_{S_0}$, because it has a component missing ($\varrho_1$) which cannot be written in terms of the other components due to $\{\varrho_{j'}\}\cup S$ being a set of linearly independent vectors. This is a contradiction, and therefore $V_S = V_{S'}$ for all such $S,S'\subset\{\varrho_j\}$, implying that $V_\varrho = V_S$. Since $\id_d\in V_S$, we have that $\id_d\in V_\varrho$. Furthermore, since $\dim(V_\varrho) = \dim(V_S) = d^2-1$, there exists at least one normalized traceless self-adjoint operator in the orthogonal subspace $V_\varrho^\perp$.
    
    Thus, in both cases, there exists at least one normalized traceless self-adjoint operator $B\in V_\varrho^\perp$. That is $\tr{B} = 0$, $\tr{B^2} = d$, and $\tr{A^*B} = \tr{AB} = 0$ for all $A\in V_\varrho$.
    
    Now, define two operators similar to \textit{Case 1} of Prop. \ref{prop:differentiability}: $\xi_\pm = \frac{1}{d}(\id_{d}\pm tB)$, where $B\in V_\varrho$, $|t|\leq\frac{1}{\sqrt{d-1}}$. Clearly, these are valid density operators, as proven in Lemma \ref{lemma1}. Furthermore, since $O\leq \xi_\pm\leq \id_{d}$, $N=\{N_+ = \frac{d}{2}\xi_+, N_- = \frac{d}{2}\xi_-\}$ is a valid measurement ($N_++N_-=\id_{d}$). With this, let us define two channels $\Phi_1,\Phi_2:\mathcal{L}(\mathcal{H})\to\mathcal{L}(\mathcal{H})$ as follows:
    \begin{align}
        \Phi_1(X) &= \tr{X}\frac{\id_{d}}{d} \\
        \Phi_2(X) &= \tr{N_+X}\xi_+ + \tr{N_-X}\xi_-.
    \end{align}
    It can be easily seen that both $\Phi_1$ and $\Phi_2$ are unital, as $\Phi_2(\id_d) = \tr{N_+\id_d}\xi_+ + \tr{N_-\id_d}\xi_- = \frac{d}{2}(\xi_+ + \xi_-) = \id_d = \Phi_1(\id_d)$. Moreover, for all $j\in[m]$, $\tr{\varrho_jN_\pm} = \frac{1}{2}\tr{\varrho_j(\id_d \pm tB)} = \frac{1}{2}\tr{\varrho_j} = \frac{1}{2}$, and therefore, for all $j\in[m]$, $\Phi_2(\varrho_j) = \frac{1}{2}(\xi_+ + \xi_-) = \frac{\id_{d}}{d} = \Phi_1(\varrho_j)$. That is, $\Phi_1$ and $\Phi_2$ cannot be differentiated by using $\{\varrho_j\}$. Therefore, $\Phi_1$ and $\Phi_2$ can be differentiated by using a set of states $\{\varrho_j\}$ and any measurement only if there exist $d^2-1$ states in $\{\varrho_j\}$ such that their associated Bloch vectors are linearly independent.
\end{proof}

The above result gives us a mathematical conditins for the set-ups to be used to differentiate between any two unital channels. A natural question is whether these conditions can be inferred directly from the set-up as in Sec. \ref{sec:sdi-identification} and what additional information we need to do so. We will see that with some additional resources we can verify the condition of Prop. \ref{prop:unitaldiff} or even the information about the unitality of the channels. We start with the following observation.

\begin{lemma}\label{lemma:fixedpt}
    Let $C,C' \in \mathcal{C}(\hi)$ such that $C = C_{\varrho,M}$ and $C' = C'_{\Phi(\varrho),M}$ for some informationally complete measurement $M$, and a channel $\Phi:\lh\to\lh$. Denote $\vec{\varrho} = (\varrho_1,\cdots,\varrho_m)^\transpose$. Then, for any non-zero vector $\vec{\alpha} = (\alpha_1,\cdots,\alpha_m)^\transpose\in\real^m$, $\vec{\alpha}\in\ker(C'^\transpose-C^\transpose)$ if and only if $\vec{\alpha}\cdot\vec{\varrho}$ is a fixed point of $\Phi$.
\end{lemma}

\begin{proof}
    Consider a real vector $\vec{\alpha} = (\alpha_1,\cdots,\alpha_m)^\transpose$. It can be seen that $(C^\transpose\vec{\alpha})_k = \sum_{j\in[m]}C_{kj}\alpha_{j} = \sum_{j\in[m]}\tr{\varrho_jM_k}\alpha_{j} = \tr{\left(\vec{\alpha}\cdot\vec{\varrho}\right)M_k}$. Similarly, $(C'^\transpose\vec{\alpha})_k = \tr{\Phi\left(\vec{\alpha}\cdot\vec{\varrho}\right)M_k}$.
    
    Now, suppose that $\vec{\alpha}\in\ker(C'^\transpose-C^\transpose)$, that is, $C'^\transpose\vec{\alpha} = C^\transpose\vec{\alpha}$. This is possible only if $\tr{\left(\vec{\alpha}\cdot\vec{\varrho}\right)M_k} = \tr{\Phi\left(\vec{\alpha}\cdot\vec{\varrho}\right)M_k}$ for all $k\in[n]$. Since $M$ is known to be informationally complete, we can conclude that $\Phi\left(\vec{\alpha}\cdot\vec{\varrho}\right) = \vec{\alpha}\cdot\vec{\varrho}$. In other words, $\vec{\alpha}\cdot\vec{\varrho}$ is a fixed point of $\Phi$.

    Conversely, let $\vec{\alpha}\cdot\vec{\varrho}$ be a fixed point of $\Phi$, that is, $\Phi\left(\vec{\alpha}\cdot\vec{\varrho}\right) = \vec{\alpha}\cdot\vec{\varrho}$. Therefore, we have that $\tr{\left(\vec{\alpha}\cdot\vec{\varrho}\right)M_k} = \tr{\Phi\left(\vec{\alpha}\cdot\vec{\varrho}\right)M_k}$ for all $k\in[n]$, which means that $(C'^\transpose\vec{\alpha})_k = (C^\transpose\vec{\alpha})_k$ for all $k\in[n]$. In other words, $C'^\transpose\vec{\alpha} = C^\transpose\vec{\alpha}$, or $(C'^\transpose- C^\transpose)\vec{\alpha} = \vec{0}$. That is, $\vec{\alpha}\in\ker(C'^\transpose-C^\transpose)$.
\end{proof}

In a practical scenario, it makes sense to study a prepare-and-measure set-up while having access to: (1) a noiseless set-up, where the states encoding the information reach the receiver without any changes, and (2) a completely noisy set-up, where none of the encoded information reaches the receiver, and what they receive is just noise. In our case, this would be equivalent to having access to the original communication matrix $C$ and the communication matrix $C_0$ in the presence of the completely depolarizing channel $\Phi_0: \lh \to \lh$ defined as $\Phi_0(X) = \frac{\tr{X}}{d} \id_d$ for all $X \in \lh$. With this additional resource, we can conclude the following.

\begin{corollary}\label{cor:unital}
    Let $C \in \mathcal{C}(\hi)$ such that $C = C_{\varrho,M}$ for some informationally complete measurement $M$. Let $C_0=C_{\Phi_0(\varrho),M}$ be the communication matrix related to the completely depolarizing channel $\Phi_0$. Then,
    \begin{itemize}
        \item[\textit{i)}] if $\ker(C_0^\transpose-C^\transpose) = \{\vec{0}\}$ and $\rank{C} \geq d^2-1$, then the set-up can be used to differentiate between any two unital channels, 
        \item[\textit{ii)}] if $\ker(C_0^\transpose-C^\transpose) \neq \{\vec{0}\}$, a channel $\Phi: \lh \to \lh$ implementing a communication matrix $C' = C_{\Phi(\varrho),M}$ is unital if and only if $\ker(C'^\transpose-C^\transpose)\cap\ker(C_0^\transpose-C^\transpose) \neq \{\vec{0}\}$.
    \end{itemize}
\end{corollary}

\begin{proof}
    \textit{i)} It is known that the only fixed points of $\Phi_0$ in $\lsh$ are real scalar multiples of $\id_d$. Thus, according to Lemma \ref{lemma:fixedpt}, any real vector $\vec{\alpha}\in\ker(C_0^\transpose-C^\transpose)$ if and only if $\vec{\alpha}\cdot\vec{\varrho} = \lambda\id_d$ for some $\lambda\in\real$. In other words, $\ker(C_0^\transpose-C^\transpose) = \{\vec{0}\}$ if and only if $\id_d\notin V_\varrho$. Thus, also for any subset $S \subseteq \{\varrho_j\}_j$ we have $\id_d \notin \rspan{S}$. If now $\rank{C}\geq d^2-1$, since we know that the measurement $M$ is informationally complete, from Prop. \ref{prop:rank} we see that $\dim(V_\varrho)\geq d^2-1$. Thus, there exists a subset $S \subseteq \{\varrho_j\}_j$ of $d^2-1$ linearly independent states such that $\id_d \notin \rspan{S}$. Together with the informational completeness of $M$ the conditions of Prop. \ref{prop:unitaldiff} are met and thus the set-up can be used to differentiate between any two unital channels.
    
    \textit{ii)} Now, suppose that $\ker(C_0^\transpose-C^\transpose) \neq \{\vec{0}\}$. As explained above, any non-zero vector $\vec{\alpha}\in\ker(C_0^\transpose-C^\transpose)$ gives us $\vec{\alpha}\cdot\vec{\varrho} = \lambda\id_d$ for some $\lambda\in\real$. This means that for any channel $\Phi$ that implements a communication matrix $C'= C_{\Phi(\varrho),M}$ using the same set-up, as per Lemma \ref{lemma:fixedpt}, we have that $\vec{\alpha} \in \ker(C'^\transpose-C^\transpose)$ if and only if $\frac{1}{\lambda}\vec{\alpha} \cdot \vec{\varrho} = \id_d$ is a fixed point of the channel $\Phi$, i.e., $\Phi$ is unital.
\end{proof}

Thus, the previous corollary shows that if we know that the measurement of the prepare-and-measure set-up is informationally complete and we have access to the completely depolarizing channel as an additional resource, then we can \textit{a)} either determine the unitality of an unknown channel, or \textit{b)} if we observe that $\rank{C}\geq d^2-1$ then we can verify that the set-up can be used to differentiate between any two unital channels, but not both. This is a form of complementarity: either we can verify a property of a channel (unitality) or we can verify that set-up is useful (for differentiation) for channels with this property.

Let us still comment on our assumption that the information about the informational completeness of the measurement is needed. Of course, if we observe that $\rank{C}=d^2$, we know (as given by Prop. \ref{prop:rank}) that both the states and the measurement are informationally complete so that the set-up can be used to differentiate between any two channels. Naturally in this instance the case \textit{i)} of Corollary \ref{cor:unital} is no longer useful, but still case \textit{ii)} can be used to determine the unitality of any channel given that we have access to the completely depolarizing channel. Thus, in this case we can indeed verify whether a given channel is unital or not just by observing that $\rank{C}=d^2$ and $\ker(C_0^\transpose-C^\transpose) \neq \{\vec{0}\}$.

However, in the case when we are merely given the information about the informational completeness of the measurement related to the set-up, we can still try to consider how to verify this information with some possible additional resources. A simple observation is that if we can prepare an informational complete set of states independent of the original set-up, then we can use it to verify the informational completeness of the measurement that is used in the set-up. Namely, if we use a set of states $\{\varrho'_l\}_l$, given a set-up consisting of states $\{\varrho_j\}_j$ and a measurement $M$, then if we see that $\rank{C_{\varrho',M}}=d^2$, we can infer that $M$ is informationally complete (as is our states $\{\varrho'_l\}_l$). Then if additionally $\rank{C_{\varrho,M}}\geq d^2-1$, we know that $\dim(V_\varrho)\geq d^2-1$ as before. So given that we have access to the additional resource of a state preparator which is capable of preparing an informationally complete set of states, then we can verify the informational completeness of the measurement. As explained above, then together with Cor. \ref{cor:unital} we can either verify the usefulness of the set-up in differentiating unital channels or verifying the unitality of an unknown channel. This shows that with some additional resources we can make weaker conclusions about the set-up even if we do not observe that $\rank{C}=d^2$ and/or $\lmax{C}=d$ as in Sec. \ref{sec:sdi-identification}.

Lastly, we see that as differentiability of channels is possible without the full informational completeness of the set-up, so is also full channel tomography. Nevertheless, for channel tomography we still need to know the full description of the set-up. Let us exemplify this in the case of unital channels. Suppose that we have an informationally complete  measurement $M \in \mathcal{M}_n(\hi)$ and a set of  states $\{\varrho_j\}_{j\in[m]}\subset\mathcal{S}(\mathcal{H})$. From the informational completeness of the measurement we get that
    \begin{equation}
        \sum_k\beta^k_bM_k=\sigma_b\qquad\forall b
    \end{equation}
   for some coefficients $\{\beta^k_b\}_{b,k} \subset \real$, and where $\{\sigma_b\}_{b=1}^{d^2}$ is an orthogonal basis consisting of self-adjoint operators with $\sigma_0=\id_d$ satisfying $\tr{\sigma_b}=d\delta_{0b}$ and $\tr{\sigma_a\sigma_b}=d\delta_{ab}$ .
    Suppose that we have a unital channel $\Phi \in \mathrm{Ch}(\hi)$ which is not the completely depolarizing channel. (In the case the unital channel is the completely depolarizing channel then according to Example \ref{ex1}, it is not necessary to know anything about the states $\{\varrho_j\}$. In fact, it is sufficient to know $\tr{M_k}$ for all $k\in[n]$ to completely characterize $\Phi$.) Now using the communication matrix $C'$ given by the states $\{\varrho_j\}_j$, the measurement $M$ and the channel $\Phi$ we get that
    \begin{equation}
    \begin{split}
        \sum_k\beta^k_bC'_{jk} &= \sum_k\beta^k_b\tr{\Phi(\varrho_{j})M_k},\\
        &= \tr{\Phi(\varrho_{j})\sum_k\beta^k_bM_k},\\
        &= \tr{\Phi(\varrho_{j})\sigma_b}.        
    \end{split}
    \end{equation}
    Expanding each $\varrho_{j}$ in the Bloch representation and applying the affine transformation of the channel $\Phi_{ba}$ ($\Phi_{0b}=d\delta_{0b}$ for valid quantum channels and $\Phi_{a0}=d\delta_{a0}$ since the channel is unital) gives us
    \begin{equation}
        \sum_k\beta^k_bC'_{jk} = \frac{1}{d}\tr{\left(\id_d + \sum_{u=1}^{d^2-1}\sum_{a=1}^{d^2-1}\Phi_{ua}r_a^{(j)}\sigma_u\right)\sigma_b} = \sum_{a=1}^{d^2-1}\Phi_{ba}r_a^{(j)},
    \end{equation}
    where $\vec{r}^{(j)} =(r^{(j)}_1, \ldots, r^{(j)}_{d^2-1})$ is the Bloch vector of the state $\varrho_j$ for all $j \in [m]$. Let us denote by $T=(\Phi_{ba})$ the $(d^2-1)\times (d^2-1)$ matrix of channel parameters, by $C=(C_{jk})$ the $m\times n$ communication matrix, by $B=(\beta^k_b)$ the $n\times (d^2-1)$ matrix of coefficients, and by $R= (r_a^{(j)})$ the $(d^2-1)\times m$ matrix with Bloch vectors as columns. The above equation becomes:
    \begin{equation}
        (CB)^\transpose = TR
    \end{equation}
    Thus, to obtain all parameters $\Phi_{ba}$, that is, the elements of the matrix $T$, the right pseudo-inverse $R^{-1}$ must exist. This is possible if and only if $R$ has a full-row rank, $d^2-1$, that is, $R$ has $d^2-1$ linearly independent columns. Since the columns of $R$ are the Bloch vectors associated with the states, we conclude that all channel parameters $\Phi_{ba}$ can be obtained if and only if there exist $d^2-1$ states in $\{\varrho_j\}$ such that the associated Bloch vectors are linearly independent. In this case we thus have that  $T = (CB)^\transpose R^{-1}$. We note that again in order to actually solve these channel parameters we must know the full descriptions of $B$ and $R$, i.e., the full descriptions of the states and the measurement.

\subsection{Entanglement-breaking channels} Lastly, let us consider the class of entanglement-breaking channels. By definition a channel $\Phi: \lh \to \lh$ is entanglement-breaking if and only if $(id \otimes \Phi)(\varrho)$ is separable for all $\varrho \in \mathcal{S}(\hi \otimes \hi)$. Equivalently, by \cite{HorodeckiShorRuskai2003}, $\Phi$ is entanglement-breaking if and only if it is of the \textit{measure-and-prepare} form, i.e., there exists some measurement $N \in \mathcal{M}_l(\hi)$ and some set of states $\{\xi_i\}_{i=1}^l \subset \hi$ such that
\begin{align}
    \Phi(X) = \sum_{i=1}^l \tr{N_i X} \xi_i
\end{align}
for all $X \in \lh$. 

Our first observation is that in the proof of Prop. \ref{prop:differentiability} both the introduced channels $\Phi_1$ and $\Phi_2$, which we could not differentiate between if the states and the measurement related to the prepare-and-measure set-up were not both informationally complete, are of the measure-and-prepare form so that they are entanglement-breaking. This means that \emph{even if we know that we are differentiating between entanglement-breaking channels, we cannot differentiate between any two such channels unless our prepare-and-measure is fully informationally complete}. Thus, in the case of differentiating between entanglement-breaking channels, or in the case of performing full channel tomography for them, the assumption of informational completeness cannot be relaxed.

On the other hand we can still ask how the entanglement-breaking property of a channel can be detected from the properties of the related communication matrix. In order to do so we need two additional quantities of communication matrices, namely the nonnegative rank and the positive semidefinite rank of a communication matrix.

Let $C$ be a nonnegative $m \times n$ matrix, meaning that all of its elements are nonnegative real numbers. Then the \emph{nonnegative factorization} of $C$ is any factorization $C=AB$, where $A$ and $B$ are nonnegative $m \times k$ and $k \times n$ matrices, respectively, for some $k \in \nat$. We say that the \emph{nonnegative rank} of $C$, denoted by $\nnrank{C}$, is the smallest $k \in \nat$ for which there exists such nonnegative factorization \cite{Cohen1993Nonnegative}. Clearly $\rank{C} \leq \nnrank{C} \leq \min(m,n)$. In \cite{HeinosaariKerppoLeppäjärvi2020} it has been shown that for a communication matrix $C$ we have $\nnrank{C}\leq k$ if and only if it can be implemented in a prepare-and-measure scenario with a $k$-dimensional classical system (equivalently $k$-dimensional quantum system with only diagonal density matrices and effects). 

A nonnegative matrix $C$ also admits a \emph{positive semidefinite factorization} as $C_{jk} = \tr{A_j B_k}$ for all $j \in [m]$ and $k \in [n]$, where $A$ and $B$ are $k \times k$ positive semidefinite complex matrices for some $k \in \nat$. We define the \emph{positive semidefinite rank (or psd-rank)} of $C$, denoted by $\psdrank{C}$, as the minimum $k \in \nat$ such that such psd-factorization exists \cite{Fawzi2015}. It can be shown that $\sqrt{\rank{C)}} \leq \psdrank{C} \leq \nnrank{C}$ for all nonnegative matrices $C$ \cite{Gouveia2013}. In \cite{FiMaPoTiWo12,LeeWeideWolf2017} it was shown that for a communication matrix $C$ we have that $\psdrank{C}\leq d$ if and only if it can be implemented in a prepare-and-measure scenario with a $d$-dimensional quantum system, i.e., $C \in \mathcal{C}(\hi)$ for a $d$-dimensional Hilbert space $\hi$. In \cite{HeinosaariKerppoLeppäjärviPlavala2024} it was shown that for all $k,d \in \nat$ there exists a communication matrix $C$ with $\psdrank{C} \leq d$ such that $\nnrank{C} \geq k$ so that the set of communication matrices implemented with a $d$-dimensional quantum theory cannot be implemented with any fixed-size classical system.

We can now state our first result about detecting entanglement-breaking channels by their communication matrices as follows:

\begin{proposition}\label{prop:ebchannel}
    A communication matrix $C' \in \mathcal{C}(\hi)$ on a $d$-dimensional Hilbert space can be realized by an entanglement-breaking channel $\Phi: \lh \to \lh$ if and only if $C'$ admits a nonnegative factorization $C' = AB$ such that
    \begin{align}
        \max(\mathrm{rank}_{psd}(A),\mathrm{rank}_{psd}(B))\leq d \, .
    \end{align}
\end{proposition}
\begin{proof}
    Let first $\Phi: \lh \to \lh$ be an entanglement-breaking channel so that there exists some states $\{\xi_i\}_{i \in [l]} \subset \sh$ and a measurement $N \in \mathcal{M}_l(\hi)$ such that $\Phi(X) = \sum_{i=1}^l \tr{N_i X}\xi_i$ for all $X \in \lh$. Let $C'\in\mathcal{C}_{m,n}(\hi)$ be a communication matrix realized by $\Phi$ so that there exists some states $\{\varrho_j\}_{j \in [m]} \subset \sh$ and a measurement $M \in \mathcal{M}_n(\hi)$ such that 
    \begin{align}
        C'_{jk} &= \tr{\Phi(\varrho_j) M_k} = \sum_{i=1}^l \tr{\varrho_j N_i} \tr{\xi_i M_k}  \label{eq:eb}
    \end{align}
    for all $j \in [m]$ and $k \in [n]$. Now we can define  nonnegative stochastic matrices $A$ and $B$ as $A_{ji} = \tr{\varrho_j N_i}$ and $B_{ik}=\tr{\xi_i M_k}$ for all $i \in [l]$, $j \in [m]$ and $k \in [n]$. Clearly $A$ and $B$ are now also quantum communication matrices, $A, B \in \mathcal{C}(\hi)$, so that by \cite{FiMaPoTiWo12,LeeWeideWolf2017} we have that $\mathrm{rank}_{psd}(A)\leq d$ and $\mathrm{rank}_{psd}(B) \leq d$. On the other hand, suppose that $C=AB$ for some nonnegative matrices $A$ and $B$ such that $\mathrm{rank}_{psd}(A)\leq d$ and $\mathrm{rank}_{psd}(B)\leq d$. By \cite{FiMaPoTiWo12,LeeWeideWolf2017} we have that $A, B \in \mathcal{C}(\hi)$ so that we can go backwards in Eq. \eqref{eq:eb} to conclude that $C'$ can be implemented with an entanglement-breaking channel in some prepare-and-measure set-up.
\end{proof}

We note that actually solving the nonnegative and the psd-factorizations are known to be NP-hard problems \cite{Vavasis2009On,Shitov2016The} so even though it may not be practical, we can in principle determine if the communication matrix can be implemented by an entanglement-breaking channel. We also note that in the previous results even if we are able to find the required nonnegative and psd-factorizations, we still cannot conclude whether a given channel is entanglement-breaking or not. However, if we know the informational completeness of the set-up, we can make conclusions also about the given channel.

\begin{corollary}\label{corr:psdrank}
     Let $C = C_{\varrho,M}\in \mathcal{C}(\hi)$ be communication matrix on a $d$-dimensional Hilbert space $\hi$ implemented by some states $\{\varrho_j\}_j$ and a measurement $M$ such that $\rank{C}=d^2$. Let $\Phi: \lh \to \lh$ be a channel implementing a communication matrix $C'= C_{\Phi(\varrho),M}$. Then $\Phi$ is entanglement-breaking if and only if $C'$ admits a non-negative factorization $C' = AB$ such that
     \begin{equation}
        \max(\mathrm{rank}_{psd}(A),\mathrm{rank}_{psd}(B))\leq d \, .
    \end{equation}
    In this case for any realization of $\Phi$ in terms of some set of states $\{\xi_i\}_{i=1}^l \subset \sh$ and a measurement $N \in \mathcal{M}_l(\hi)$ for some $l \in \nat$ such that $\Phi(X) = \sum_{i=1}^l \tr{N_i X}\xi_i$ for all $X \in \lh$ we have that
          \begin{align}
          \min(\dim(V_{\xi}), \dim(V_N)) &\geq \rank{C'}\label{eq:ebrank1}\\
          \max(\dim(V_{\xi}), \dim(V_N)) &\leq l\label{eq:ebrank2},
          \end{align}
    and $\min(\dim(V_{\xi}), \dim(V_N)) = \rank{C'}$ whenever $\max(\dim(V_{\xi}), \dim(V_N)) = l$.
\end{corollary}

\begin{proof}
    Given a communication matrix $C = C_{\varrho,M}\in \mathcal{C}(\hi)$ such that $\rank{C} = d^2$, which means that $\dim(V_\varrho) = \dim(V_M) = d^2$ according to Corollary \ref{corr:infocomplete}. Now, suppose that we have a channel $\Phi:\lh\to\lh$ implementing the communication matrix $C'= C_{\Phi(\varrho),M}$. According to Prop. \ref{prop:ebchannel}, there exists an entanglement-breaking channel $\Psi$ implementing the same communication matrix if and only if $C'$ admits a non-negative factorization $C' = AB$ such that $\max(\mathrm{rank}_{psd}(A),\mathrm{rank}_{psd}(B))\leq d$. Since $\rank{C} = d^2$, it is possible to perfectly differentiate between any two channels using $C$, according to Prop. \ref{prop:differentiability}. Since $\Phi$ and $\Psi$ give rise to the same communication matrix, this means that $\Phi=\Psi$. Therefore, all implementations of $C'$ are entanglement-breaking if and only if $C'$ admits a non-negative factorization $C' = AB$ such that $\max(\mathrm{rank}_{psd}(A),\mathrm{rank}_{psd}(B))\leq d$.

    Let $\Phi$ be such an entanglement breaking channel. Then, $\Phi$ can be realized using some set of states $\{\xi_i\}_{i=1}^l \subset \sh$ and a measurement $N \in \mathcal{M}_l(\hi)$ for some $l \in \nat$ such that $\Phi(X) = \sum_{i=1}^l \tr{N_i X}\xi_i$ for all $X \in \lh$. This gives us the nonnegative factorization $C'=AB$ such that $A_{ji} = \tr{\varrho_jN_i}$ and $B_{ik} = \tr{\xi_iM_k}$. Since the rank of product of two matrices is less than the smallest rank of the two, we have that $\rank{C'} \leq \min(\rank{A},\rank{B})$. Since both $A$ and $B$ are communication matrices and since $\dim(V_\varrho) = \dim(V_M) = d^2$, according to Prop. \ref{prop:rank}, we have that $\dim(V_N) = \rank{A}$ and $\dim(V_\xi) = \rank{B}$. Therefore, we have that $\min(\dim(V_{\xi}), \dim(V_N)) \geq \rank{C'}$. Furthermore, since $V_\xi$ and $V_N$ are spanned by the $l$ states and the $l$ effects, we must have that $\max(\dim(V_{\xi}), \dim(V_N)) \leq l$.

    Now, according to the rank-nullity theorem, for any $m\times l$ matrix $A$ and $l\times n$ matrix $B$, we have
    \begin{equation}\label{eq:rnt}
        \rank{AB} = \rank{B} - \dim(\ker(A)\cap\im(B)).
    \end{equation}
    Suppose that $\max(\dim(V_{\xi}), \dim(V_N)) = l$. That is, either $\rank{B} \leq \rank{A} = l$ or $\rank{A} \leq \rank{B} = l$. If $\rank{A} = l$, then $\ker(A) = \{0\}$, and by Eq. \eqref{eq:rnt}, we have that $\rank{C'} = \rank{AB} = \rank{B} = \min(\dim(V_{\xi}), \dim(V_N))$. On the other hand, if $\rank{A} \leq \rank{B} = l$, then $\ker(A)\cap\im(B) = \ker(A)$ and by Eq. \eqref{eq:rnt}, we have $\rank{C'} = l - \dim(\ker(A)) = \rank{A} = \min(\dim(V_{\xi}), \dim(V_N))$.
\end{proof}

We also note that $\ker(A) \neq \{0\}$ if and only if there exists some non-zero vector $\vec{u}\in\real^l$ such that $(A\vec{u})_j = \sum_{i\in[l]}u_i\tr{\varrho_jN_i} = \tr{\varrho_j\left(\sum_{i\in[l]}u_iN_i\right)} = 0$ for all $j\in[m]$. Since $\{\varrho_j\}$ is informationally complete, this is possible if and only if the effects of the measurement $N$ are linearly dependent, which can happen if and only if the measurement $N$ is non-extremal \cite{Watrous_2018}. In other words, the equality of Eq.~\eqref{eq:ebrank1} holds whenever $N$ is an extremal measurement.

However, Corollary \ref{corr:psdrank} still has to be stated as inequalities in general, since there exist cases where both inequalities are strict. We demonstrate this below with an example.

\begin{example}
    Let $d=2$ and the prepare-and-measure set-up be given by the six-outcome measurement $M = \{M_k^\pm = \frac{1}{6}(\id_2 \pm \sigma_k)\}_{k\in\{x,y,z\}}$ and the six states $\varrho_j^\pm = 3M_j^\pm$ for $j\in\{x,y,z\}$. The communication matrix before introducing any channel is given by,
    
    \begin{equation}
        C = \frac{1}{6}\begin{pmatrix}2 & 0 & 1 & 1 & 1 & 1\\0 & 2 & 1 & 1 & 1 & 1\\1 & 1 & 2 & 0 & 1 & 1\\1 & 1 & 0 & 2 & 1 & 1\\1 & 1 & 1 & 1 & 2 & 0\\1 & 1 & 1 & 1 & 0 & 2\end{pmatrix}.
    \end{equation}
    Now, consider the measure-and-prepare channel $\Phi:\mathcal{L}(\complex^2)\to\mathcal{L}(\complex^2)$ realized by four-outcome measurement $N = \{N_i^\pm = \frac{1}{4}(\id_2 \pm \sigma_i)\}_{i\in\{x,y\}}\}$ and the four states $\{\xi_i^\pm = \frac{1}{2}(\id_2 + \sigma_i)\}_{i\in\{x,y\}}\}\}$. Clearly $\xi_i^+ = \xi_i^-$. The non-negative factors of $C'$, given by $A = \left(\tr{\varrho_j^\pm N_i^\pm}\right)$ and $B = \left(\tr{\xi_i^\pm M_k^\pm}\right)$, will then become
    \begin{equation}
         A= \frac{1}{4}\begin{pmatrix}2 & 0 & 1 & 1\\0 & 2 & 1 & 1\\1 & 1 & 2 & 0\\1 & 1 & 0 & 2\\1 & 1 & 1 & 1\\1 & 1 & 1 & 1\end{pmatrix},\;B = \frac{1}{6}\begin{pmatrix}2 & 0 & 1 & 1 & 1 & 1\\2 & 0 & 1 & 1 & 1 & 1\\1 & 1 & 2 & 0 & 1 & 1\\1 & 1 & 2 & 0 & 1 & 1\end{pmatrix}.
    \end{equation}
    It can be easily seen that $\rank{A} = 3$ and $\rank{B} = 2$. Furthermore, the communication matrix in the presence of the channel, $C'$, is given by
    \begin{equation}
        C' = \frac{1}{12}\begin{pmatrix}3 & 1 & 3 & 1 & 2 & 2\\3 & 1 & 3 & 1 & 2 & 2\\3 & 1 & 3 & 1 & 2 & 2\\3 & 1 & 3 & 1 & 2 & 2\\3 & 1 & 3 & 1 & 2 & 2\\3 & 1 & 3 & 1 & 2 & 2\end{pmatrix},
    \end{equation}
    which is a single point contraction according to Example \ref{ex1}, and has $\rank{C'} = 1$. Therefore, we have that $\rank{C'} < \min(\rank{A},\rank{B}) < \max(\rank{A},\rank{B}) < l$, which shows that the inequalities in Corollary \ref{corr:psdrank} can be strict.
\end{example}

To summarize, we see that if we are using an informationally complete set-up, the conditions of Prop. \ref{prop:ebchannel} become necessary and sufficient for a given channel to be entanglement-breaking. Additionally in this case from the rank of the communication matrix $C'$ after the application of the channel we also get information about the measure-and-prepare form of the entanglement-breaking channel. Namely, we see that the minimal dimension of the measure-and-prepare implementation as well as the number of states and the number of outcomes of the measurement must be at least the rank of $C'$. In particular, in the special case when we detect that $\rank{C'}=d^2$ we can conclude that any measure-and-prepare realization of the channel must be informationally complete.

\section{Conclusions}\label{sec:concl}

In this work we have extended the communication matrix framework and included the quantum processes connecting preparations and measurements. Using this framework, we have examined in detail all the prerequisites for channel differentiation and channel tomography. While channel differentiation is possible with only the knowledge of the informational completeness of the set-up, for full channel tomography we need to have the full description of the set-up. We show that the information completeness can be readily checked from the rank of the communication matrix related to the set-up. Furthermore, we show that for the most practical prepare-and-measure scenarios, we can self-test the set-up up to an orthogonal transformation by calculating the information storability of the communication matrix. We can then use this information about the description of the set-up to identify quantum channels. 

Subsequently we have looked into scenarios where some additional information and/or resources are provided about the channels, which are used to certify that the set-up can characterize certain classes of channels or deduce certain properties of the channels. In particular we looked into the case where we can use the channel multiple times which did not help with differentiability between all channels. On the other hand we saw that with the additional resources of having access to the completely depolarizing channel and an informationally complete set of test states, we can either verify the unitality of an unknown channel or use the set-up to differentiate between any two unital channels. Similarly we were able to show that in principle also the entanglement-breaking form of a channel can be verified by calculating the nonnegative and psd-factorizations of the communication matrix. Our list of channel properties is by no means exhaustive and one could for example look into more finer properties of qubit channels which are much more studied. One could also combine some of the studied cases and look whether for example using the channel multiple times could help with verifying some properties of the channels.

The present framework considers channel identification in a prepare-and-measure scenario, where the channel acts on input states of the system alone. This is distinct from the most general channel-discrimination setting, in which the input may be entangled with an ancillary system and the output can be measured jointly with that ancilla. Ancilla-assisted strategies can provide potentially better discrimination capabilities. Formally, such strategies can be incorporated by replacing the channel $\Phi$ with $\Phi\otimes  id_A$ and treating the joint system as the effective input/output system; however, doing so changes the dimensionality and the assumptions concerning the accessible preparations and measurements. Consequently, the results presented here should be understood as characterizing channel identification within the considered semi-device-independent prepare-and-measure architecture, rather than as establishing optimality with respect to unrestricted ancilla-assisted discrimination. Extending the communication-matrix approach to explicitly exploit ancillary systems is an interesting direction for future work.

Although we try to focus our results on quantum theory for the majority of this work, it is noteworthy to mention that the inequalities in Prop. \ref{prop:rank} should also translate to other operational theories, as the inequalities can be derived using the same methods and the appropriate inner product (as opposed to the Hilbert-Schmidt inner product). However, the same cannot be said for Prop. \ref{prop:implement}, since maximal information storability is not sufficient to certify the set-up in all operational theories to our knowledge.

\section*{Author contributions}
SP and LL participated in the investigation of the problem, the formalization of the propositions and proofs, the preparation of the first draft, the organization of the sections, and the revisions. SP prepared the figures and developed the examples. MZ supervised the work, provided conceptual guidance, and made the final corrections.

\section*{Acknowledgements}
We thank Nikolai Miklin, Mariami Gachechiladze and Oskari Kerppo for their helpful comments in identifying the correct gauge in Prop.  \ref{prop:implement}. This work is supported by DeQHOST APVV-22-0570. SP is supported by skQCI DIGITAL-2021-QCI-01 No. 101091548 and QUAS VEGA 2/0164/25. LL is supported by the Business Finland project BEQAH (Between Quantum Algorithms and Hardware). MZ is supported by QENTAPP 09I03-03-V04-00777.

\bibliographystyle{quantum}
\bibliography{Bib}

\end{document}